\documentclass[AEJ, draftmode]{AEA_Sebi}%
\usepackage{longtable}
\usepackage{amsfonts}
\usepackage[round]{natbib}
\usepackage{amsmath}
\usepackage{amssymb}
\usepackage{graphicx}
\usepackage{hyperref}
\usepackage[margin=0.95in]{geometry}%
\providecommand{\U}[1]{\protect\rule{.1in}{.1in}}

\newtheorem{assumption}{Assumption}

\newtheorem{lemma}{Lemma}

\newtheorem{proposition}{Proposition}

\begin{document}

\title{A Social Network Analysis of Occupational Segregation}
\author{I. Sebastian Buhai and Marco J. van der Leij\thanks{Buhai (corresponding author) is affiliated with
SOFI\ at Stockholm University, NIPE at Minho University, and CEPREMAP in Paris;
email:\ \href{mailto:sbuhai@gmail.com}{sbuhai@gmail.com}. Van der Leij
is affiliated with the Congregation of the Blessed Sacrament and the University of Amsterdam;
email:\ \href{mailto:mvanderleij@gmail.com}{mvanderleij@gmail.com}. The latest
version can be downloaded from
\url{https://www.sebastianbuhai.com/papers/publications/segregation_networks.pdf}.
We thank Editor-in-Chief Thomas Lubik, an anonymous Associate Editor, and especially two anonymous referees for detailed recent feedback.
For comments and suggestions on (many) previous drafts, we are grateful to
Willemien Kets, Leeat Yariv, Mich\`{e}le Belot, Ben Golub, Sanjeev Goyal, Wilbert Grevers,
Maarten Janssen, Joan de Mart\'{\i}, Friederike Mengel, James Montgomery,
Jos\'{e} Luis Moraga-Gonz\'{a}lez, Wojciech Olszewski, Gilles Saint-Paul, Ott
Toomet, Jan van Ours, Fernando Vega-Redondo, and Yves Zenou. We
also acknowledge audiences from seminars at Stockholm University, Babes-Bolyai
University, Northwestern University, CPB\ Netherlands Bureau for Economic
Policy Analysis, Aarhus University, Tartu University, Tinbergen Institute,
University College London,\ and from numerous workshops and conferences.}}
\date{\ December 2022}

\begin{abstract}
\pagenumbering{arabic} 
We propose an equilibrium interaction model of occupational segregation
and labor market inequality between two social groups, generated exclusively through the documented tendency
to refer informal job seekers of identical ``social color''. The expected social color homophily
in job referrals strategically induces distinct career choices for
individuals from different social groups, which further translates into stable partial
occupational segregation equilibria with sustained wage and employment
inequality -- in line with observed patterns of racial or gender labor market disparities.
Supporting the qualitative analysis with a calibration and simulation exercise, we furthermore show that 
both first and second best utilitarian social optima entail segregation, 
any integration policy requiring explicit distributional concerns. Our framework highlights 
that the mere social interaction through homophilous contact networks can be a pivotal
channel for the propagation and persistence of gender and racial
labor market gaps, complementary to long studied mechanisms such as taste or
statistical discrimination. \medskip



\end{abstract}
\maketitle

\shortTitle{A Social Network Analysis of Occupational Segregation}
\baselineskip=18pt plus 2pt \lineskip=3pt minus 1pt \lineskiplimit=2pt

\pagebreak 
\section{Introduction}

Most studies investigating the causes of labor market inequality
agree that classical theories such as taste or statistical discrimination by
employers cannot, \textit{alone,} explain pay, employment, and occupational
disparities between genders or races, and their remarkable persistence over
time.\footnote {There have been, for instance, countless empirical studies within both sociology and economics
documenting the extent and shape of occupational segregation -- a consistent finding being concisely summed up by Richard Posner: \textquotedblleft a glance of the composition
of different occupations shows that in many of them, particularly racial,
ethnic, and religious groups, along with one or the other sex and even groups
defined by sexual orientation (heterosexual vs. homosexual), are
disproportionately present or absent\textquotedblright%
. The quote is from the last paragraph of an essay published
by Posner on \textquotedblleft The\ Becker-Posner Blog",\ on 30-01-2005
(address:
\url{https://www.becker-posner-blog.com/2005/01/larry-summers-and-women-scientists--posner.html}
last retrieved on  20-09-2022). Posner goes on by illustrating with an example 
of gender-based occupational segregation that is less likely
to be due to discrimination:\ \textquotedblleft a much higher percentage of
biologists than of physicists are women, and at least one branch of biology,
primatology, appears to be dominated by female scientists.\textquotedblright%
} While several meritorious complementary theories have been
advanced, some leading social scientists have suggested
that social interactions could also be an important, yet relatively little
explored channel in this context, see for instance \citet{arrow}.\footnote{For more recent overviews of potential channels explaining
observed labor market inequality between genders, see for instance
\citet{Bertrand}, \citet{GoldinKatz2011, GoldinKatz2016}, or
\citet{BlauKahn2017}.}

In this paper, we investigate a potential network channel leading to
occupational segregation and wage inequality in the labor market, by
developing and analysing an intuitive, parsimonious social interactions model.\footnote{The role of informal personal networks for inter-gender
labor market inequality had been emphasized in sociology, often as part of the
gender-specific "social capital", at least since \citet{Burt1992,Burt1998}. In
economics the interest took up more slowly, and was little welcome for quite a while -- e.g., a first discussion version of this paper, containing already part of our current analysis/ results, was public in 2006, see \url{https://papers.tinbergen.nl/06016.pdf}, when few economists were working in this direction and even some hostility to the idea could be perceived -- but, optimistically, there has recently been quite a wave of published studies on the role of personal networks for gender (and/or racial, ethnical) disparities on the
labor market, using a diverse set of approaches, see, e.g.,
\citet{zeltzer}, \citet{Mengel 2020}, \citet{LindenlaubPrummer}, as well as other relevant studies referenced in the comprehensive recent review by \citet{Jackson 2022}.}
We construct a four-stage model of occupational segregation between two
homogeneous, exogenously given, mutually exclusive social groups (e.g.,
genders or races) acting in a two-job labor market. 
While the model stages are formally described and detailed in Section \ref{secsebi3}, we sketch them intuitively below. In the first stage each
individual chooses one of two specialized educations to become a worker. In
a second stage individuals randomly form "friendship"\ ties with other
individuals, with a tendency to form relatively more ties with members of the
same social group, what is known in the literature as \textquotedblleft%
\textit{(inbreeding) homophily}\textquotedblright, \textquotedblleft%
\textit{inbreeding bias}\textquotedblright\ or "\textit{assortative
matching}".\footnote{Homophily measures the relative frequency of within-group
versus between-group friendships. There exists inbreeding homophily or an
inbreeding bias if the group's homophily is higher than what would have been
expected if friendships are formed randomly. See, e.g.,
\citet{Currarini et al} for formal definitions.} In a third stage
individuals search for jobs, either directly or through their networks of employed friendship contacts. In the
last stage workers earn a wage and spend their income on a single
consumption good. 

We obtain the following results. First, unsurprisingly, we show that with
inbreeding homophily within social groups, a complete polarization in terms of
occupations across the two groups arises as a stable equilibrium outcome. This
follows from standard arguments on network effects. If a group is
completely segregated and specialized in one type of job, then each individual
in the group has many more job contacts if she "sticks" to her specialization.
Hence, sticking to one specialization ensures good job opportunities to group
members, and these incentives stabilize segregation.

We next extend the basic model allowing for \textquotedblleft
good\textquotedblright\ and \textquotedblleft bad\textquotedblright\ jobs, in
order to analyze equilibrium wage and unemployment inequality between the two
social groups. We show that with large differences in job attraction (i.e.,
wages at equal labor supply), the main outcome of the model is that one social
group "fully specializes" in the good job, while the other group "mixes" over
the two jobs. In this partial segregation equilibrium, the group that
specializes in the good job always has a higher payoff and a lower
unemployment rate.\footnote{Throughout the paper, "payoff" (or "expected payoff") is short for "ex ante expected payoff".}
Furthermore, with a sufficiently large intra-group
homophily, the fully-specializing group also has a higher equilibrium
employment rate \emph{and} a higher wage rate than the "mixing" group, thus
being twice advantaged. Hence, our model is able to explain typical empirical
patterns of gender, race, or ethnic labor market inequality — see, e.g., our next (sub)section reviewing known empirical facts on gender segregation patterns in the labour market. The driving force
behind our result is the fact that the group that fully specializes, being
homogenous occupationally, is able to create a denser job contact network than
the mixing group. We emphasize at this point that we do not intend to imply that there is no more taste or statistical discrimination by
employers in the labor market.\ On the contrary, we regard our social
interaction model, classical discrimination theories, as well as other
theories such as, e.g., gender-specific work amenity preferences, as
\textit{complementary} bits in explaining observed patterns of labor market inequality.

We finally consider whether society benefits from an integration policy, in
the sense that labor inequality between the social groups would be attenuated.
To this aim, we analyze a social planner's first and second-best policy
choices.\footnote{Intuitively, our first best utilitarian social policy assumes that the social planner can fully control the social group fractions choosing one or the other education track (for instance, by being able to force switches). In our second-best social welfare analysis, the social planner allows individual incentives to shape the educational choices, having only the means to stabilize a symmetric 'integration' equilibrium, in which the fraction of individuals choosing a particular education is the same in both social groups. More formal/ detailed descriptions appear later, in the corresponding section containing the social welfare analyses.} 
Rather counter-intuitively, segregation is the preferred outcome in the first-best
analysis and a laissez-faire policy leading to segregation shaped by
individual incentives is always maximizing social welfare in the second-best case.
Hence, overall employment is higher under segregation, while laissez-faire
wage inequality remains sufficiently small, such that segregation is an
overall socially optimal policy. We show that integration policies are
justified only in the presence of additional distribution concerns, beyond
individual utilities. Our social welfare analysis points out therefore some policy
issues unfortunately ignored in most debates concerning anti-segregation legislature,
such as the need for \textit{the explicit promotion of distributional concerns}.

Our model shares similarities with the frameworks by \citet{benabou} and
\citet{Mailath et al}. \citet{benabou} introduces a model in which individuals
choose between high and low education; the benefits of education are
determined globally, but the costs are determined by local education
externalities. As in our model, these local education externalities lead to
segregation and inequality at the macro level. Unlike in our model, inequality
in education level fully explains pay gaps in \citet{benabou}, which, for
instance, is at odds with the evidence on gender education gaps --- see the
discussion in our Section \ref{secsebi1_occeg}. \citet{Mailath et al} also
consider a model in which workers choose between high and low (no) education,
but in a setting with search and matching between firms and workers. The
features of their and our segregation equilibria are similar, even though the
behavioral mechanisms behind them are very different. Crucially, in their
model firms may target their search to one group of workers; they show that
there exists a segregation equilibrium, in which workers of only one group
invest in high skills and firms target their search to this highly-skilled
group. In that case, unemployment is lower for the highly-skilled group, whose
wage is also higher, since lower unemployment gives them a better bargaining
position vis-\`{a}-vis the firms. In both \citet{benabou} and
\citet{Mailath et al} social welfare policies are ambiguous: depending on the
parameters, either integration or segregation might be socially optimal. This
is in stark contrast to our model, in which a segregation is always a first
best optimal policy, and under reasonable assumptions also a second-best
optimal policy. It suggests that ignoring the channel of homophilous job
contacts may overestimate the welfare effects of integration policies. While
we believe the mechanisms considered by \citet{benabou} and
\citet{Mailath et al} are present as well (and so are other channels,
including direct taste discrimination), we show in our paper that neither
spatially local education externalities, nor targeted firm search are needed
to explain the wage and employment gaps between gender or racial groups: job
search through social networks combined with inbreeding homophily is already
sufficient. We provide a small calibration and simulation exercise showing that our model
generates wage and unemployment gaps in line with actual figures.

Significant progress has been achieved in modeling labor market phenomena by
means of social networks. Existing research has for instance investigated the
effect of social networks on employment, wage inequality, and labor market
transitions.\footnote{The seminal paper on the role of networks in the labor
markets is \citet{Montgomery 1991}. Other well cited papers include, e.g.,
\citet{calvojackson1, calvojackson2}, \citet{IoannidesSoutevent}, or
\citet{bramoullesp};\ for more general reviews on the economic analysis of
social networks see,e.g., \citet{Goyal 2007,Goyal 2016},
\citet{Jackson 2008}.} This work points out that individual
performance on the labor market crucially depends on the position individuals
take in the social network structure. However, these studies typically do not
focus on the role that networks play in accounting for persistent patterns of
occupational segregation and inequality between races, genders or
ethnicities.\footnote{\citet{calvojackson1} find that two groups with two
different networks may have different employment rates due to the endogenous
decision to drop out of the labor market. However, their finding draws heavily
on an example that already assumes a large amount of inequality; in
particular, the groups are initially unconnected and the initial employment
state of the two groups is unequal.} A recent exception is
\citet{Bolteetal} who model and discuss how the distribution of job referrals could lead
to persistent inequality and intergenerational immobility, also discussing the impact on productivity. While we do not 
analyse the wider impact and implications of the \textit{network structure} or the intergenerational dynamic aspects, our simpler, 'reduced form' approach
enables us to focus on and highlight the essential role of the homophilous referral mechanism that relates indirect job finding via
social network contacts to occupational segregation and labor market inequality in earnings and employment
between social groups.\footnote{Linked to the intergenerational focus of \citet{Bolteetal}, it is perhaps of interest to note that our model could easily be enhanced to also address some of the relevant dynamic, intergenerational aspects in this context of labour markets with homophilous job referrals, among which the persistence of the occupational segregation over time, for instance by analyzing a myopic best-reply dynamics starting from given initial conditions (this specific implementation has been suggested to us by an anonymous referee). Inter alia, the gender or racial homophily in job referrals could thus be explained to contribute essentially to the perpetuation of occupational segregation and labour market inequalities along gender and race lines.}
In particular, we use our model to explore strategic career choices in this context of homophilous job contact networks, complementing therefore the scope of \citet{Bolteetal}.

The next section overviews empirical findings on labor market gaps between genders/ races,
on the relevance of job contact networks, and on the extent of social group
homophily. We set up our model of occupational segregation in
Section~\ref{secsebi3} and discuss key results on the segregation
equilibria in Section~\ref{sec_equilibrium}. Section \ref{sec_welfare}
analyses social welfare optima. We summarize and conclude in
Section~\ref{secsebi6}.

\section{Empirical background}

\label{secsebi2}

In this section we review the empirical background that motivates the
building blocks of our model. We first discuss evidence on occupational
segregation, and the relation to gender and race wage gaps. Next, we overview
empirical literature on the role of job contact networks and on homophily.

\subsection{Wage gaps and segregation by occupation and education fields,
between genders and races}

\label{secsebi1_occeg}

Although labor markets have become more open to traditionally disadvantaged
groups, wage differentials by race and gender remain stubbornly persistent,
see, e.g., \citet{altonjiblank},
\citet{BlauKahn2000, BlauKahn2006, BlauKahn2017}, \citet{England et al}, or \citet{Jackson 2022}.
\citet{altonjiblank} note for instance that in 1995 a full-time employed white
male earned on average \$ 42,742, whereas a full-time employed black male
earned on average \$ 29,651, thus 30 \% less, and an employed white female \$
27,583, that is, 35 \% less. Although the pay gaps diminished over time,
especially between genders, in 2018 US women still earned only 83 \% of what
men did at the median, see, e.g., \citet{England et al}. Standard wage
regressions are typically able to explain only half of this earnings gap, but
more detailed analysis reveals more insights. In particular, several authors
have found that the inclusion of individual Armed Forces Qualifying Test
scores is able to almost fill the wage gap on race. On the one hand, this
suggests that the pay gap between whites and blacks might be largely created
before individuals enter the labor market, though complementary channels
cannot be fully excluded. On the other hand, it is clear that the gender wage
gap cannot be fully accounted for by pre-market factors, as women have
nowadays caught up and even got ahead men in terms of education levels, see,
e.g., \citet{BlauKahn2017}, \citet{England et al};\ for a more detailed
analysis on the difficulty of explaining the gender pay gap with both typical
and many atypical observables, see, e.g., \citet{ManningSwaffield}.

Much research within social sciences suggests that segregation into separate
type of jobs, i.e. occupational segregation, explains a large part of the
gender wage gap, as well as part of the race wage gap. One major candidate for
the residual wage gap must be the fact that women are underrepresented in
higher-paying, typically male-dominated occupations: see, e.g.,
\citet{MacphersonHirsch}, \citet{Bayard et al}, \citet{Goldin 2014}. There are plenty of studies presenting detailed \textit{direct}
statistics on the occupational segregation\footnote{Some of these papers, e.g.,
\citet{Sorensen 2004}, discuss in detail the extent of labor market
segregation between social groups, at the \textit{workplace, industry and
occupation }levels. Here we are concerned with modeling segregation by
\textit{occupation} alone (known also as "horizontal segregation").
Remark that  our stylized model may also explain particular inter-industry disparities, 
but we emphasize the occupational dimension, since this appears to be dominant relative 
to segregation by industry: e.g., \citet{WeedenSorensen}\ show that occupational segregation in the
USA\ is much stronger than segregation by industry and that if one wishes to
focus on one single dimension, \textquotedblleft occupation is a good choice,
at least relative to industry\textquotedblright.} and wage inequality
patterns by gender, race or ethnicity, see, e.g., \citet{beller},
\citet{albelda}, \citet{King 1992}, \citet{PadavicReskin}, \citet{charles},
\citet{Cotter et al}, \citet{BlauKahn2017}, \citet{England et al}. They all
agree that, despite substantial expansion in the labor market participation of
women and affirmative action programs aimed at labor integration of racial and
ethnic minorities, women typically remain clustered in female-dominated
occupations, while blacks (as well as some other races or ethnic groups) are
over-represented in some occupations and under-represented in others. Moreover, the occupations where women or blacks segregate
are usually\ of lower 'quality', meaning \textit{inter alia }that
they are paying less on average, which partly explains the male-female and
white-black wage differentials.

\citet{King 1992} offers, for instance, detailed evidence that throughout
1940-1988 there was a persistent and remarkable level of occupational
segregation by race and sex, such that \textquotedblleft approximately
two-thirds of men or women would have to change jobs to achieve complete
gender integration\textquotedblright, with some changes in time for some
subgroups. Whereas occupational segregation between white and black women
appears to have diminished during the 60's and the 70's, occupational
segregation between white and black males or between males and females
remained remarkable stable. Several studies by Barbara Reskin and her
co-authors, c.f. the discussion and references in \citet{PadavicReskin},
document the extent of occupational segregation by narrow race-sex-ethnic
cells and find that segregation by gender remained extremely prevalent and
that within occupations segregated by gender, racial and ethnic groups are
also aligned along stable segregation paths. \citet{England et al} summarize
the trends in the segregation of occupations by means of an occupational
segregation index ranging from 0 (complete integration) to 1 (complete
segregation), which "has fallen steadily since 1970[...] moving from 0.60 to
0.42. However, it moved much faster in the 1970s and 1980s than it has since
1990; segregation dropped by 0.12 in the 20-y period after 1970, but by a much
smaller 0.05 in the longer 26-y period after 1990." Though most of these
studies are for the USA, there is also international evidence confirming that,
with some variations, similar patterns of segregation hold, e.g.,
\citet{PettitHook} or, especially, \citet{EIGE 2020} -- which concludes among others
that although the gap between male and female employment rates in EU countries has
narrowed, labor markets remain highly segregated, about 30 \% of employed
women working in sectors relatively poorly paid and of lower work quality,
such as education, health or social work.\footnote{The exact wording in
\citet{EIGE 2020}:\ "The progress on women's participation
has not led to substantial changes to gendered patterns of employment in the
labour market. The Index score for work quality and segregation has scarcely
changed since 2010, standing at 64 points in 2018. Around 30 \% of all employed 
worked in education, health and social work activities, compared with 8 \% of men. 
Other sectors and occupations remain dominated by men: for example,
only 17 \% of ICT specialists are women".}

The recent studies by \citet{England et al} on US\ exhaustive CPS data and
respectively \citet{EIGE 2020} on EU-level data are in fact doubly relevant
for our purpose. Next to demonstrating, as already mentioned above, that
despite some progress, the convergence on gender pay gaps, employment gaps and
occupational desegregation has slowed down, they also show that the gender
desegregation by fields of study has stalled --- and that despite that women
by now overtook men in terms of both bachelor and doctoral degrees. Using a
desegregation index for fields of study across genders, similar to the one for
occupations, \citet{England et al} show that "[for bachelor degrees] it fell
from 0.47 in 1970 to 0.28 in 1998, and has not gone down since, but rather,
segregation has risen slightly. For doctoral degrees, segregation went from
0.35 in 1970 to a low of 0.18 in 1987 and has hovered slightly higher since.
In neither case has there been any net reduction in segregation for over 20
years." As the authors of that study also recognize, this is extremely
important, as "segregated fields of study contribute to occupational gender
segregation." At the same time, one of \citet{EIGE 2020}'s conclusions is
very pertinently worded as follows: "Gender segregation in education remains a
major barrier to gender equality in the EU. In 2017, 43 \% of all women at
university were studying education, health and welfare, humanities or the
arts, with the gender gap in the EU as a whole standing at 22 p.p., unchanged
since 2010. This divide is mirrored by gender segregation in the labour
market, which determines women's and men's earnings, career prospects and
working conditions." In the main part of our paper, we provide a mechanism of
exactly how that can happen and how that can then lead to persistent
inequality in wages and employment.

\subsection{Job contact networks}

\label{secsebi2_jcn}

There is by now an established set of facts showing the importance of the
informal job networks in matching job seekers to vacancies. For instance, on
average about 50 \% of the workers obtain jobs through their personal
contacts, e.g. \citet{Rees 1966}, \citet{Granovetter 1995},
\citet{Holzer 1987}, \citet{Montgomery 1991}, \citet{Topa 2001};
\citet{bewley} enumerates several studies published before the 90's, where the
fraction of jobs obtained via friends or relatives ranges between 30 
and 60 \%.\footnote{The difference in the use of informal job networks among
professions is also documented. \citet{Granovetter 1995}\ pointed out that
although personal ties are relevant in job search-match for all
professions, their incidence is higher for blue-collar (50 to 65
\%)\ than for white-collar categories such as accountants or typists (20
to 40 \%). However, for other white-collars the use of
social connections in job finding is even higher, e.g., as
high as 77 \% for academics.} It is also established that on
average 40-50 \% of the employers actively use social networks of their
current employees to fill their job openings, e.g. \citet{Holzer 1987}.
Furthermore, employer-employee matches obtained via contacts appear to have
some common characteristics. Those who found jobs through personal contacts
were on average more satisfied with their job, e.g., \citet{Granovetter 1995},
and were less likely to quit, e.g. \citet{Datcher}, 
\citet{SimonWarner}, \citet{DatcherLoury}. For more detailed overviews of
studies on job information networks, see \citet{IoannidesDatcher}\ or
\citet{Topa 2011}; for more recent empirical research on the influence and
value of job referral networks, see, e.g., \citet{Bayer et al},
\citet{Hellerstein et al}, \citet{Burksetal}, or the very recent review by \citet{Jackson 2022}.

\subsection{Intra-group homophily}

\label{secsebi2_hom}

There is considerable evidence on the existence of social \textquotedblleft
homophily\textquotedblright,\footnote{The "homophily theory" of
friendship was first introduced and popularized by sociologists
\citet{LazarsfeldMerton}, with \citet{Coleman58} introducing "inbreeding
homophily" indices, and the notion extensively used in sociology ever since.
In economics, the notion got popular much later, in the second half of the
2000s, but since then the literature grew massively, see, e.g, the discussion and references in the very recent papers by \citet{Bolteetal} or \citet{Jackson 2022} } also labeled \textquotedblleft assortative
matching\textquotedblright\ or \textquotedblleft inbreeding social
bias\textquotedblright, that is, there is a higher probability of establishing
links among people with similar characteristics. Extensive research shows that
people tend to be friends with similar others, see, e.g., \citet{McPherson et al} for
a good review, with characteristics such as race, ethnicity or gender being
essential dimensions of homophily. Friendship
patterns appear to be more homophilous than would be expected by chance or availability
constraints \textit{even after controlling for the unequal distribution of races or
sexes through social structure}, extensively and repeatedly documented at least since \citet{Shrum et al}. 

In our "job information network" context, early studies by \citet{Rees 1966}
and \citet{Doeringer and Piore} showed that workers who had been asked for
references concerning new hires were in general very likely to refer people
"similar" to themselves. While such similar features could be ability, education,
age, race and so on, the focus here is on groups
stratified along exogenous characteristics (i.e., one is born in such a group
and cannot alter her group membership) such as those divided along gender,
race or ethnicity lines. Indeed, most subsequent evidence on homophily was in
the context of such 'exogenously given' social groups. For instance,
\citet{Marsden 1987} finds using the U.S. General Social Survey that personal
contact networks tend to be highly segregated by race, while other studies
such as \citet{Brass85} or \citet{Ibarra},\ using cross-sectional single firm
data, find significant gender segregation in personal networks.\ More recent
evidence on various homophilous social networks is also given by
\citet{MayerPuller}, \citet{Currarini et al}, or Zeltzer (2020).

Direct evidence of large gender homophily within\textit{ job contact networks
}comes, for instance, from tabulations by \citet{Montgomery 1992}. Over all occupations in a
US sample from the National Longitudinal Study of Youth, 87 \% of the
jobs men obtained through contacts were based on information received from
other men and 70 \% of the jobs obtained informally by women were as
result of information from other women. Montgomery shows that these outcomes
hold even when looking at each narrowly defined occupation categories or
one-digit industries,\footnote{\citet{WeedenSorensen} estimate a two-dimensional model
of gender segregation, by industry \textit{and} occupation: they find much
stronger segregation across occupations than across industries.\ 86 \% of the
total association in the data is explained by the segregation along the
occupational dimension; this increases to about 93 \% once industry segregation
is also accounted for.} including traditionally male or female
dominated occupations, where job referrals for the minority group members were
obtained still with a very strong assortative matching via their own gender
group. For example, in male-dominated occupations such as machine operators,
81 \% of the women who found their job through a referral, had a female
reference. Such figures are surprisingly large and are likely to be only lower
bounds for magnitudes of inbreeding biases within other social
groups.\footnote{The gender homophily is likely to be smaller than race or
ethnic homophily, given frequent close-knit relationships between men and
women.\ This is confirmed for instance by \citet{Marsden 1988}, who finds
strong inbreeding biases in contacts between individuals of the same race or
ethnicity, but less pronounced homophily within gender categories.}

Another direct evidence for our purpose is the study by
\citet{Fernandez and
Sosa}, who use a dataset documenting both the recruitment and
the hiring stages for an entry-level job at a call center of a large
US bank. This study also finds that contact networks contribute to the gender
skewing of jobs, in addition documenting directly that there is strong
evidence of gender homophily in the refereeing process:\ referees of both
genders tend to strongly produce same sex referrals.

Some more recent studies are, for instance, by \citet{Beaman et al} who ask job applicants to a research organization in Malawi to refer other candidates and find that among referred applicants female applicants are more likely than male applicants to have been referred by a woman; and by \citet{Brown et al} who analyze data on applicants and hires at a U.S. corporation, finding that
referrers and referrals tend to be similar in terms of age, gender, ethnicity and education. In their data, 64 \% of the referral matches are between people of the same gender. Furthermore, an interesting and very recent study by \citet{Hederos et al} estimates a measure of gender homophily in job referrals, largely unaffected by the behavior of the referred candidates and the outcome of the hiring process, by directly observing the referral process of Business students in a Stockholm higher-education establishment. Inter alia, they find strong evidence of gender homophily in job referrals, with 73 \% of participants referring a candidate of their own gender — and that homophily being equally strong within the men and women student subgroups. 

Finally, we want to briefly address the relative importance of homophily within
"exogenously given" versus "endogenously created" social groups. As mentioned
above, assortative matching takes place along a great variety of dimensions.
However, there is also an empirical literature suggesting that homophily within
exogenous groups such as those divided by race, ethnicity, gender, and- to a
certain extent- religion, typically outweighs assortative matching within
endogenously formed groups such as those stratified by educational, political
or economic lines. E.g., \citet{Marsden 1988} finds for US strong inbreeding
bias in contacts between individuals of the same race or ethnicity and less
pronounced homophily by education level. Another study by
\citet{Tampubolon},\ using UK data, documents the dynamics of friendship as
strongly affected by gender, marital status and age, but not by education, and
only marginally by social class. These facts partly motivate why our focus here is on
"naturally" arising social groups, such as gender, racial or ethnic
ones. Nevertheless, as will become clear in the modeling, \textit{allowing
assortative matching by education, in addition to gender, racial or
ethnic homophily, does not matter at all for our results and conclusions}.

\section{A model of occupational segregation}

\label{secsebi3}

Based on the stylized facts mentioned in Section \ref{secsebi2}, we build a
parsimonious theoretical model of social network interaction able to explain
stable occupational segregation and employment and wage gaps, that can
complement existing theories and thus potentially account for the remaining unexplained disparities in
labor market outcomes between genders or races.

Let us consider the following setup. A continuum of individuals with measure 1
is equally divided into two social groups, Reds ($R$) and Greens ($G$). The
individuals are ex ante homogeneous apart from their social color. They can
work in two occupations, $A$ or $B$. Each occupation requires a corresponding
thorough specialized education (career track), such that a worker cannot work
in it unless she followed that education track. We assume that it is too
costly for individuals to follow both educational tracks. Hence, individuals
have to choose their education track before they enter the labor
market.\footnote{For example, graduating high school students may face the
choice of pursuing a medical career or a career in technology. Both choices
require several years of expensive specialized training, and this makes it
unfeasible to follow both career tracks.}

Consider now the following order of events:

\begin{enumerate}
\item Individuals choose one education in order to specialize in either
occupation $A$ or $B$;

\item Individuals randomly establish \textquotedblleft
friendship\textquotedblright\ relationships, thus forming a network of contacts;

\item Individuals participate in the labor market. Individual $i$ obtains a
job with probability $s^{i}$.

\item Individuals produce a single good for their firms and earn a wage
$w^{i}$. They obtain utility from consuming goods that they buy with their wage.
\end{enumerate}

We proceed with an elaboration of these steps.

\subsection{Education strategy and equilibrium concept}

The choice of education in the first stage involves strategic behavior.
Workers choose the education that maximizes their expected payoff given the
choices of other workers, and we therefore look for a Nash equilibrium in this
stage. This can be formalized as follows.

Denote by $\mu_{R}$ and $\mu_{G}$ the fractions of Reds and respectively
Greens that choose education $A$. It follows that $1-\mu_{X}$ of group
$X\in\{R,G\}$ chooses education $B$. The payoffs will depend on these
strategies:\ the payoff of a worker of group $X$ that chooses education $A$ is
given by $\Pi_{A}^{X}(\mu_{R},\mu_{G})$, and mutatis mutandis, $\Pi_{B}%
^{X}(\mu_{R},\mu_{G})$. Define $\Delta\Pi^{X}\equiv\Pi_{A}^{X}-\Pi_{B}^{X}$.
The functional form of the payoffs is made more specific later, in Subsection
\ref{wages_payoffs}.

In a Nash equilibrium each worker chooses the education that gives her the
highest payoff, given the education choices of all other workers. Since
workers of the same social group are homogenous, a Nash equilibrium implies
that if some worker in a group chooses education $A$ ($B$), then no other
worker in the same group should strictly prefer education $B$ ($A$). Hence, we define a
pair $(\mu_{R},\mu_{G})$ an \emph{equilibrium} if and only if, for
$X\in\{R,G\}$, the following hold:\footnote{The question whether the
equilibrium is in pure or mixed strategies is not relevant, because the player
set is a measure of identical infinitesimal individuals (except for group
membership). Our equilibrium could be interpreted as a Nash equilibrium in
pure strategies; then $\mu_{X}$ is the measure of players in group $X$
choosing pure strategy $A$. The equilibrium could also be interpreted as a
symmetric Nash equilibrium in mixed strategies; in that case the common
strategy of all players in group $X$ is to play $A$ with probability $\mu_{X}%
$. A hybrid interpretation is also possible.}
\begin{align}
\Delta\Pi^{X}(\mu_{R},\mu_{G})\leq0  &  \mbox{ if }\mu_{X}=0 \label{eqcondmu0}%
\\
\Delta\Pi^{X}(\mu_{R},\mu_{G})=0  &  \mbox{ if }0<\mu_{X}<1\label{eqcondmu}\\
\Delta\Pi^{X}(\mu_{R},\mu_{G})\geq0  &  \mbox{ if }\mu_{X}=1.
\label{eqcondmu1}%
\end{align}

To strengthen the equilibrium concept, we restrict ourselves to \textit{stable
}equilibria. We use a simple stability concept based on a standard myopic
adjustment process of strategies, which takes place before the education
decision is made. That is, we think of the equilibrium as the outcome of an
adjustment process. In this process, individuals repeatedly announce their
preferred education choice, and more and more workers revise their education
choice if it is profitable to do so, given the choice of the other
workers.\footnote{One could think of such a process as the discussions
students have before the end of the high school about their preferred career.
An alternative with a longer horizon is an overlapping generations model, in
which the education choice of each new generation partly depends on the choice
of the previous generation.} Concretely, we consider stationary points
of a dynamic system guided by the differential equation $\dot{\mu}_{X}%
=k\Delta\Pi^{X}(\mu_{R},\mu_{G})$. We define $\mu\equiv(\mu_{R},\mu_{G})$ a
\textit{stable equilibrium} if and only if it is an equilibrium and (i) for
$X\in\{R,G\}$: $\partial\Delta\Pi^{X}/\partial\mu_{X}<0$ if $\Delta\Pi^{X}=0$;
(ii) det$(D\Delta\Pi(\mu))>0$ if $\Delta\Pi^{R}=0$ and $\Delta\Pi^{G}=0$,
where $D\Delta\Pi(\mu)$ is the Jacobian of $(\Delta\Pi^{R},\Delta\Pi^{G})$
with respect to $\mu$.

We assumed that individuals first choose an education, and then form a network
of job contacts (see the next subsection). As a consequence, individuals have to
make \emph{expectations} about the network they could form, and base their
education decisions on these expectations. This is in contrast to some of the
earlier work on the role of networks in the labor market. In that research,
the network was supposed to be already in place, or the network was formed in
the first stage (\citet{Montgomery 1991}, \citet{calvo}, \citet{calvojackson1}).

Our departure from the earlier frameworks raises legitimate questions about
the assumed timing of the education choice. Are crucial career decisions made
before or after job contacts are formed? One might be tempted to answer: both.
Of course everyone is born with family ties, and both in early school and in the
neighborhood children form more ties, etc. It is also known that peer-group
pressure among children has a strong effect on decisions to, for instance,
smoke or engage in criminal activities and, no doubt, family and early friends
do form a non-negligible source of influence when making crucial career
decisions. However, we argue that most \textit{job-relevant contacts} (the so
called 'instrumental ties') are made later, for instance at the\ university,
or early at the workplace, hence after a specialized career track had been
chosen. In spite of the fact that those ties are typically not as strong as
family ties, they are more likely to provide relevant information on vacancies
to job seekers; \citet{Granovetter 1973,Granovetter 1995} and much subsequent literature provide convincing
evidence that job seekers more often receive crucial job information from
acquaintances ("weak ties"),\ rather than from early family or close childhood friends
("strong ties"). If a majority of instrumental ties are formed
after the individual embarked on a (irreversible)\ career, then it is
justified to consider a model in which the job contact network is formed after
making a career choice.

\subsection{Network formation}

In the second stage, individuals form a network of contacts. We assume this
network to be random, but with \textit{social color homophily}. That is, we assume that
the probability for two workers to create a tie is $p\geq0$ when the
workers are from different social groups and follow different education
tracks; however, when the workers are from the same social group, the
probability of creating a tie increases with $\lambda>0$. Similarly, if two
workers choose the same education, then the probability of creating a tie
increases with $\kappa\geq0$. Hence, \textit{we also allow for assortative matching by
education}, in addition to the one by social color. We do not impose any
further restrictions on these parameters, other than securing $p+\lambda
+\kappa\leq1.$ This leads to the tie formation probabilities from
Table~\ref{tab:tieformation}. We shall refer to two workers that create a tie
as \textquotedblleft friends\textquotedblright.

\begin{table}[ptb]
\caption{Probability of a tie between two individuals, depending on group membership and education choice.}%
\label{tab:tieformation}
\begin{center}%
\begin{tabular}
[c]{|lr|cc|}\hline\hline
&  & \multicolumn{2}{c|}{Education}\\
&  & same & different\\\hline
Social & same & $p+\kappa+\lambda$ & $p+\lambda$\\
group &  &  & \\
& different & $p+\kappa$ & $p$\\\hline\hline
\end{tabular}
\end{center}
\end{table}

We assume the probability that an individual $i$ forms a tie with individual
$j$ to be exogenously given and constant. In practice, establishing a
friendship between two individuals typically involves rational decision
making. It is therefore plausible that individuals try to optimize their job
contact network in order to maximize their chances on the labor
market.\footnote{See \citet{calvo} for a model of strategic network
formation in the labor market.} In particular, individuals from the
disadvantaged social groups should have an incentive to form ties with
individuals from the advantaged group. While this argument is probably true,
we do not incorporate this aspect of network formation in our model. The
reality is that strategic network formation does not appear to dampen the
inbreeding bias in social networks significantly;\ in Section~\ref{secsebi2_hom}
we have provided evidence that strong homophily exists even within groups that have
strong labor market incentives \emph{not} to preserve such homophily in
forming their ties. The reason could be that the payoff of forming a tie is
mainly determined by various social and cultural factors, and only for a
smaller part by benefits from the potential transmission of valuable job
information.\footnote{\citet{Currarini et al} discuss a model of network
formation in which individuals form preferences on the number and mix of
same-group and other-group friends. In this model inbreeding homophily arises
endogenously.} On top of that, studies such as, e.g.,
\citet{Granovetter 2002}, also note that many people would feel exploited if
they found out that someone befriended them for the selfish reason of obtaining
job information. These elements might hinder the role of labor market
incentives when forming ties. Hence, while we do not doubt that incentives do
play a role when forming ties, we believe that they are not sufficiently strong
to undo the effects of the social color homophily. We thus consider that endogenizing network formation
in this particular setup would not alter--while needlessly obscuring--the gist of our current analysis.

\subsection{Job matching and social networks}

The third stage we envision for this model is that of a dynamic labor process,
in which information on vacancies is propagated through the social network, as
in, e.g., \citet{calvojackson1}, \citet{calvozenou},
\citet{IoannidesSoutevent} or \citet{bramoullesp}. Workers who randomly lose
their job are initially unemployed because it takes time to find information
on new jobs. The unemployed worker receives such information either directly,
through formal search, or indirectly, through employed friends who receive the
information and pass it on to her (in the particular case where all her
friends are unemployed, only the formal search method works). As the specific
details of such a process are irrelevant for our purposes, we do not
consider the dynamic aspects explicitly, taking a "reduced form" approach that highlights our model's main mechanisms.

In particular, we assume that unemployed workers have a higher propensity to
receive job information when they have more friends with the \emph{same job
background}, that is, with the \emph{same choice of education}. On the one
hand, this assumption is based on the result of \citet{IoannidesSoutevent}
that in a random network setting the individuals with more friends have a
lower unemployment rate.\footnote{This result is nontrivial, as the unemployed
friends of employed individuals tend to compete with each other for job
information. Thus, if a friend of a jobseeker has more friends, the
probability that this friend passes information to the jobseeker decreases. In
fact, in a setting in which everyone has the same number of friends,
\citet{calvozenou} show that the unemployment rate is non-monotonic in the
(common) number of friends. Remark also that our modeling abstracts from employment probability non-monotonicity concerns potentially arising from limited total supply of jobs in an occupation. It is not clear that such a constraint is extremely relevant in practice, however.} On the other hand, this assumption is
based on the conjecture that workers are more likely to receive information
about jobs in their own occupation. For example, when a vacancy is opened in a
team, the other team members are the first to know this information, and are
also the ones that have the highest incentives to spread this information around.

Formally, denote the probability that individual $i$ becomes employed by
$s^{i}=s(x_{i})$, where $x_{i}$ is the measure of friends of $i$ with the same
education as $i$ has. We thus assume that $s(x)$ is differentiable, $0<s(0)<1$
(there is non-zero amount of direct job search) and $s^{\prime}(x)>0$ for all
$x>0$ (the probability of being employed increases in the number of friends
with the same education).

It is instructive to show how $s^{i}$ depends on the education choices of $i$
and the choices of all other workers. Remember that the population is equally divided in Reds and Greens and that $\mu_{R}$ and $\mu_{G}$
are the fractions of Reds and respectively Greens that choose education $A$.
Given the tie formation probabilities from Table~\ref{tab:tieformation} and
some algebra, the employment rate $s_{A}^{X}$ of $A$-educated workers in social group~$X\in
\{R,G\}$ will be given by:
\begin{equation}
s_{A}^{X}(\mu_{R},\mu_{G})=s\left(  (p+\kappa)\bar{\mu}+\lambda\mu
_{X}/2\right)  \label{eqsax}%
\end{equation}
and likewise, the employment rate $s_{B}^{X}$ of $B$-educated workers in social group~$X$ will
be
\begin{equation}
s_{B}^{X}(\mu_{R},\mu_{G})=s\left(  (p+\kappa)(1-\bar{\mu})+\lambda(1-\mu
_{X})/2\right)  \label{eqsbx}%
\end{equation}
where $\bar{\mu}\equiv(\mu_{R}+\mu_{G})/2$.

Note that $s_{A}^{X}>s_{A}^{Y}$ and $s_{B}^{X}<s_{B}^{Y}$ for $X,Y\in\{R,G\}$,
$X\neq Y$, if and only if $\mu_{X}>\mu_{Y}$ and $\lambda>0$. We will see in
Section~\ref{secsebi4} that the ranking of the employment rates is crucial, as
it creates a group-specific network effect. That is, keeping this ordering, if
only employment matters (jobs are equally attractive), then individuals have
an incentive to choose the same education as other individuals in \emph{their}
social group. Importantly, it is straightforward to see that this ordering of
the employment rates depends on $\lambda$, but it does not depend on $\kappa$.
Thus, only the homophily among members of the same social group -- and not
the assortative matching by education -- is relevant for our results.

\subsection{Wages, consumption and payoffs}

\label{wages_payoffs}

The eventual payoff of the workers depends on the wage they receive, the goods
they buy with that wage, and the utility they derive from consumption. Without
loss of generality we assume that an unemployed worker receives zero wage.
However, the wages of employed workers are not exogenously given, but they are
determined by supply and demand.

When firms offer wages, they take into account that there are labor market
frictions and that it is impossible to employ all workers simultaneously. Thus
what matters is the effective supply of labor as determined by the labor
market process in stage 3. Let $L_{A}$ be the total measure of employed
$A$-educated workers and $L_{B}$ be the total measure of employed $B$-educated workers. Hence,
\begin{equation}
L_{A}(\mu_{R},\mu_{G})=\mu_{R}s_{A}^{R}(\mu_{R},\mu_{G})/2+\mu_{G}s_{A}%
^{G}(\mu_{R},\mu_{G})/2 \label{eqla}%
\end{equation}
and
\begin{equation}
L_{B}(\mu_{R},\mu_{G})=(1-\mu_{R})s_{B}^{R}(\mu_{R},\mu_{G})/2+(1-\mu
_{G})s_{B}^{G}(\mu_{R},\mu_{G})/2. \label{eqlb}%
\end{equation}
Given (\ref{eqsax}) and (\ref{eqsbx}) from above, it is easy to check that
$L_{A}$ is increasing with $\mu_{R}$ and $\mu_{G}$, whereas $L_{B}$ is
decreasing with $\mu_{R}$, $\mu_{G}$.

As in \citet{benabou}, consumption, prices, utility, the demand for labor and
the implied wages are determined in a 1-good, 2-factor general equilibrium
model. All individuals have the same utility function $U:R_{+}\rightarrow R$,
which is strictly increasing and strictly concave with $U(0)=0$. The single
consumer good sells at unit price, such that consumption of this good equals
wage and indirect utility is given by $U_{i}=U(w_{i})$.

Whereas in \citet{benabou} firms combine high- and low-skilled workers, here firms put $A$-educated and $B$-educated workers together to produce the single good at
constant returns to scale. Wages are then determined by the production
function $F(L_{A},L_{B})$. As usually, we assume that $F$ is strictly
increasing and strictly concave in $L_{A}$ and $L_{B}$ and $\partial
^{2}F/\partial L_{A}\partial L_{B}>0$. Writing the wage as function of
education choices and using (\ref{eqla}) and (\ref{eqlb}), the wages of
$A$-educated and $B$-educated workers, $w_{A}$ and $w_{B}$, are given by
\[
w_{A}(\mu_{R},\mu_{G})=\frac{\partial F}{\partial L_{A}}\left(  L_{A}(\mu
_{R},\mu_{G}),L_{B}(\mu_{R},\mu_{G})\right)  ,
\]
and
\[
w_{B}(\mu_{R},\mu_{G})=\frac{\partial F}{\partial L_{B}}\left(  L_{A}(\mu
_{R},\mu_{G}),L_{B}(\mu_{R},\mu_{G})\right)  .
\]
It is easy to check that $w_{A}$ is strictly decreasing with $\mu_{R}$ and
$\mu_{G}$, and mutatis mutandis, $w_{B}$.

We can now define the payoff of a worker as her expected utility at the time
of decision-making. The payoff function of an $A$-educated worker from social
group $X\in\{R,G\}$ is thus
\begin{equation}
\Pi_{A}^{X}(\mu_{R},\mu_{G})=s_{A}^{X}(\mu_{R},\mu_{G})U(w_{A}(\mu_{R},\mu
_{G})). \label{pi_r_a}%
\end{equation}
Similarly,
\begin{equation}
\Pi_{B}^{X}(\mu_{R},\mu_{G})=s_{B}^{X}(\mu_{R},\mu_{G})U(w_{B}(\mu_{R},\mu
_{G})). \label{pi_r_b}%
\end{equation}

If we do not impose further restrictions, then there could be multiple
equilibria, most of them uninteresting. To ensure a unique equilibrium in our
model (actually: two symmetric equilibria), we make the following two assumptions.

\begin{assumption}
For the wage functions $w_{A}$ and $w_{B}$
\[
\lim_{x\downarrow0} U(w_{A}(x,x))= \lim_{x\downarrow0} U(w_{B}
(1-x,1-x))=\infty.
\]
\label{asswage}
\end{assumption}

\begin{assumption}
For $X \in\{R,G\}$, and for all $\mu_{R}, \mu_{G}\in[0,1]$
\[
\left|  \frac{\partial s_{A}^{X}/s_{A}^{X}}{\partial\mu_{X}/\mu_{X}}\ \right|
< \left|  \frac{\partial U/U}{\partial w_{A}/w_{A}} \right|  \left|
\frac{\partial w_{A}/w_{A}}{\partial\mu_{X}/\mu_{X}} \right|
\]
and
\[
\left|  \frac{\partial s_{B}^{X}/s_{B}^{X}}{\partial\mu_{X}/\mu_{X}}\ \right|
< \left|  \frac{\partial U/U}{\partial w_{B}/w_{B}} \right|  \left|
\frac{\partial w_{B}/w_{B}}{\partial\mu_{X}/\mu_{X}} \right|  .
\]
\label{assemployment}
\end{assumption}

Assumptions \ref{asswage} and \ref{assemployment} guarantee the uniqueness of
our results. Assumption~\ref{asswage} implies that the wage for scarce labor
is so high that at least some workers always find it attractive to choose
education $A$ or respectively $B$; everyone going for one of the two
educations cannot be an equilibrium. In Assumption \ref{assemployment} we
assume that the education choice of an individual has a smaller marginal
effect on the employment probability within a group than on the wages and
overall utility. Note that the assumption implies that for $X\in\{R,G\}$
\[
\frac{\partial\Pi_{A}^{X}}{\partial\mu_{X}}<0<\frac{\partial\Pi_{B}^{X}%
}{\partial\mu_{X}},
\]
and it is this feature that guarantees the uniqueness of our results. The
assumption is not restrictive as long as there is sufficient direct job
search, because the employment probability of each individual in our model is
bounded between $s(0)>0$ and $1$, with $s(0)$ capturing the employment
probability in the absence of any ties and thus induced only by the
exogenously given direct job finding rate. Hence, a higher $s(0)$ implies less
of an impact of the network effect on the employment rate.

Remark that we make these assumptions above only to focus
our analysis on segregation outcomes, for the sake of clarity and brevity.
\textit{These assumptions are not necessary}. For instance, in the calibration from
Section~\ref{secsim}, Assumption 2 is violated, but there are still (two)
unique equilibria.

\section{Equilibrium results}

\label{sec_equilibrium}

We now present the equilibrium analysis of our model. Formal proofs for all
subsequent propositions are relegated to the "Proofs" appendix. Without loss of
generality we assume throughout the section that $w_{A}(1,0)\geq w_{B}(1,0)$,
thus that the $A$-occupation is weakly more attractive than the $B$-occupation
when effective labor supply is the same. We call $A$ the \textquotedblleft
good\textquotedblright\, and $B$ the \textquotedblleft
bad\textquotedblright\ job.

\subsection{Occupational segregation}

\label{secsebi4}

We are in particular interested in those equilibria in which there is
segregation. We define \emph{complete segregation} if $\mu_{R}=0$ and $\mu
_{G}=1$, or, vice versa, $\mu_{R}=1$ and $\mu_{G}=0$. On the other hand, we
say that there is \emph{partial segregation} if for $X\in\{R,G\}$ and
$Y\in\{R,G\}$, $Y\neq X$: $\mu_{X}=0$ but $\mu_{Y}<1$, or, vice versa,
$\mu_{X}=1$ but $\mu_{Y}>0$.

Our first result is that segregation, either complete or partial, is the only
stable outcome:

\begin{proposition}
\label{propsegregation} Suppose Assumptions~\ref{asswage} and
\ref{assemployment} hold. Define $s_{H} \equiv s((p+\kappa+\lambda)/2)$ and
$s_{L} \equiv s((p+\kappa)/2)$.

\begin{enumerate}
\item[(i)] If
\begin{equation}
1 \leq\frac{U(w_{A}(1,0))}{U(w_{B}(1,0))} \leq\frac{s_{H}}{s_{L}},
\label{prop3cond1}%
\end{equation}
then there are exactly two stable equilibria, both with complete segregation.

\item[(ii)] If
\begin{equation}
\frac{U(w_{A}(1,0))}{U(w_{B}(1,0))} > \frac{s_{H}}{s_{L}}, \label{prop3cond2}%
\end{equation}
then there are exactly two stable equilibria, both with partial segregation,
in which either $\mu_{R}=1$ or $\mu_{G}=1$.
\end{enumerate}
\end{proposition}

We first note that a non-segregation equilibrium cannot exist, even in the
case of a tiny amount of homophily $\lambda$. The intuition is that homophily
in the social network among members of the same social group creates a
group-dependent network effect, that is, the utility of a specific education to an individual will increase with the number of individuals from the same social group that pick that same education. This network externality implies, for instance, that if slightly more Red workers choose $A$
than Greens do, then the value of an $A$-education is higher for the Reds than
for the Greens, while the value of a $B$-education is lower in the Reds'
group. Positive feedback then ensures that the initially small differences in
education choices between the two groups widen and widen, until at least one
group segregates completely into one type of education.

Second, if the utility ratio, and hence (since utility is strictly increasing in the wage) the wage ratio or, to make it most intuitive, the wage differential between the two jobs (for equal supply of A-educated and B-educated workers) $w_{A}(1,0)-w_{B}(1,0)$, is not "too large" vis-\`{a}-vis the social
network effect (condition \ref{prop3cond1}), complete segregation is the only
stable equilibrium outcome, given a positive inbreeding bias in the social
group. Thus one social group specializes in one occupation, and the other
group in the other occupation. On the other hand, the proposition makes clear
that complete segregation cannot be sustained if the wage differential is "too
large" vis-\`{a}-vis the social network effect (condition \ref{prop3cond2}).
Starting from complete segregation, a large wage differential gives incentives
to the group specialized in $B$-jobs to switch to $A$-jobs.

Interestingly, the "unsustainable" complete segregation equilibrium is then
replaced by a partial segregation equilibrium in which a group specializes in the
\textquotedblleft good\textquotedblright\ job $A$, while the other has
both $A$-educated and $B$-educated workers. Partial segregation in which one group, say the
Greens, fully specializes in the \textquotedblleft bad\textquotedblright\ job
$B$ is unsustainable, as it would lead to an oversupply of $B$-educated workers and
an even larger wage gap. This would provide Red $B$-educated workers with
strong incentives to switch en masse to the $A$-job.

\subsection{Inequality}

\label{secsebi5}

The discussion so far ignored eventual equilibrium differentials in wages and
unemployment between the two types of jobs. We now tackle that case. We
continue to assume that $w_{A}(1,0)\geq w_{B}(1,0)$ and, in light of the
results of Proposition \ref{propsegregation}, we focus without loss of
generality on the equilibrium in which $\mu_{R}=1$. Thus, the Reds specialize
in the \textquotedblleft good\textquotedblright\ job $A$, while the
\textquotedblleft bad\textquotedblright\ job $B$ is only performed by Green workers.

We first consider the case in which wage differentials are small enough so
that complete segregation is an equilibrium ($\mu_{R}=1$ and $\mu_{G}=0$). In
this case the implications are straightforward. Since both groups specialize
in equal amounts, the network effects are equally strong, and the employment
rates are equal. Given that employment rates are equal, effective labor
supply is also equal, and therefore the wage of the \textquotedblleft
good\textquotedblright\ job is weakly higher. We thus have the following result:

\begin{proposition}
\label{prop4b} Suppose Assumptions~\ref{asswage} and \ref{assemployment} hold.
Define $s_{H}\equiv s((p+\kappa+\lambda)/2)$ and $s_{L}\equiv s((p+\kappa)/2)$
and suppose that $1 \leq\frac{w_{A}(1,0)}{w_{B}(1,0)}\leq\frac{s_{H}}{s_{L}}$.
Then $(\mu_{R}, \mu_{G})=(1,0)$ is a stable equilibrium, in which
\[
w_{A} \geq w_{B},
\]
\[
s_{A}^{R}=s_{B}^{G}>s_{B}^{R}=s_{A}^{G},
\]
and
\begin{equation}
\Pi_{A}^{R} \geq\Pi_{B}^{G} \geq\Pi_{A}^{G} \geq\Pi_{B}^{R}.
\label{prop3comppayoff}%
\end{equation}

\end{proposition}

Next, we turn to the analysis of the more interesting case in which wage
differentials for equal labour supply are large. In that case there is a partial segregation equilibrium in which
$(\mu_{R},\mu_{G})=(1,\mu^{\ast})$ where $\mu^{\ast}\in(0,1)$. First note that
according to (\ref{eqcondmu}) this implies the following condition:
\[
\Pi_{A}^{G}(1,\mu^{\ast})=\Pi_{B}^{G}(1,\mu^{\ast}),
\]
or equivalently
\[
s_{A}^{G}(1,\mu^{\ast})U(w_{A}(1,\mu^{\ast}))=s_{B}^{G}(1,\mu^{\ast}%
)U(w_{B}(1,\mu^{\ast})).
\]
Thus, whereas workers in group $R$ prefer the $A$-job, workers in group
$G$ make an individual trade-off: lower wages should be exactly compensated by
higher employment probabilities and vice versa.

We are particularly interested in whether this individual trade-off between
unemployment and wages translates into a similar trade-off at the
group level, in that an equilibrium inter-group wage gap could also be "compensated" by a reversed
employment gap. We have the following proposition.

\begin{proposition}
\label{prop6} Suppose Assumptions~\ref{asswage}~and~\ref{assemployment} hold.
Define $s_{H}\equiv s((p+\kappa+\lambda)/2)$ and $s_{L}\equiv s((p+\kappa)/2)$
and suppose that $\frac{w_{A}(1,0)}{w_{B}(1,0)}>\frac{s_{H}}{s_{L}}$. Define
$\hat{\mu}\in(0,1)$, such that
\begin{equation}
w_{A}(1,\hat{\mu})=w_{B}(1,\hat{\mu}), \label{marginalworker}%
\end{equation}
and let $(\mu_{R},\mu_{G})=(1,\mu^{\ast})$ be a stable equilibrium. In that
equilibrium
\begin{equation}
\Pi_{A}^{R}>\Pi_{B}^{G}=\Pi_{A}^{G}>\Pi_{B}^{R}. \label{prop5comppayoff}%
\end{equation}
Moreover,

\begin{enumerate}
\item[(i)] if $\hat{\mu} < \frac{\lambda}{2(p+\kappa+\lambda)}, $ then
\[
s_{A}^{R}>s_{B}^{G}>s_{A}^{G}>s_{B}^{R},
\]
and
\[
w_{A}(1,\mu^{*}) > w_{B}(1,\mu^{*});
\]

\item[(ii)] if $\hat{\mu} > \frac{\lambda}{2(p+\kappa+\lambda)}, $ then
\[
s_{A}^{R}>s_{A}^{G}>s_{B}^{G}>s_{B}^{R},
\]
and
\[
w_{B}(1,\mu^{*}) > w_{A}(1,\mu^{*}).
\]

\end{enumerate}
\end{proposition}

The main implication of this proposition is that an equilibrium inter-group wage gap
\emph{might not always be} compensated by a reversed employment gap. It is thus
possible that the group specializing in the good job, here the Reds, both
earns a higher wage \emph{and} has higher employment probabilities than the
Greens. Indeed, this is what happens with a large initial differential
in attractiveness between good and bad jobs and, when group homophily bias $\lambda$
is large relative to $p$ and $\kappa$ (in fact $p+\kappa$), hence when we are in situation (i) described above. 
A "trade-off" between the equilibrium inter-group wage and employment gaps (where, in fact, a reversed equilibrium inter-group wage gap compensates now for the employment gap) only occurs 
for $\lambda$ sufficiently small relative to $p+\kappa$, hence in case (ii) from above.

These results can be intuitively understood by making the following observations. First, the workers in the
'specializing' group $R$ have a higher employment probability than \emph{all}
workers in group $G$. This is always the case, regardless of whether the
individual in $G$ is an $A$ or a $B$ worker, and whether $s_{B}^{G}>s_{A}^{G}$
or not. As all members of group $R$ choose the same occupation, the Reds
remain a strong homogenous social group. Network formation with homophily then
implies that they are able to create a lot of ties, and hence, that they
benefit most from their social network. On the other hand, the Greens are
dispersed between two occupations. This weakens their social network, decreasing
their chances on the labor market, for both $A$ and $B$-educated workers belonging to the
Greens group. Second, whether the equilibrium wage differential between the workers in the two groups is
positive or negative depends on the relative size of $\lambda$ relative to
$p+\kappa$, in the term $\frac{\lambda}{2(p+\kappa+\lambda)}$ from the
inequality conditions in Proposition \ref{prop6}. This can be roughly assessed
in light of the empirical evidence on homophily discussed earlier in this
paper. First, as seen from the stylized facts from Section \ref{secsebi2_hom}, the
assortative matching by education, $\kappa$, is typically found to be weaker
relative to racial, ethnical or gender homophily. The second interesting
situation is a scenario where the probability of making contacts in general,
$p$, were already extremely high relative to the intra-group homophily bias.
However, given the surprisingly large size of intra-group inbreeding biases in
personal networks of contacts found empirically, this is also unlikely. Hence,
the likelihood is very high that in practice $\lambda$ would dominate the
other parameters in the cutoff term $\frac{\lambda}{2(p+\kappa+\lambda)}$ (and thus in practice we would in the world corresponding to case (i) from above).

Let us now summarize all the implications of the latest Proposition. The fully
specializing group is always better off in terms of unemployment rate and
payoff, independent of either relative or absolute sizes of $\lambda$, $p$ and
$\kappa$ (as long as $\lambda>0$), as shown in Proposition \ref{prop6}.
Furthermore, given the observed patterns of social networks discussed in
Section~\ref{secsebi2_hom}, the condition of $\lambda$ dominant relative to $p$
and $\kappa$ is likely to be met. This ensures that the group fully
specializing in the good job always has a higher wage in the equilibrium than
the group mixing over the two jobs, as proven in Proposition \ref{prop6}. Note
that this partial segregation equilibrium is in remarkable agreement with
observed occupational, wage and unemployment disparities in the labor market
between genders or races, as for instance overviewed in Subsection \ref{secsebi1_occeg} for gender disparities. This suggests that our
model offers a plausible explanation for major empirical patterns of
labor market inequality.

\section{Social welfare}

\label{sec_welfare}

\subsection{First best social optimum}

In the previous section we obtained that individual incentives lead to
occupational segregation and to wage and unemployment inequality. Could this imply
that a policy targeting integration may reduce inequality, and in fact
may just be socially beneficial -- argument often used
by proponents of positive discrimination? We thus seek to also analyze the
implications of our model from a social planner's point of view.

Consider a utilitarian social welfare function:
\begin{equation}
W(\mu_{R},\mu_{G})=\mu_{R}\Pi_{A}^{R}/2+(1-\mu_{R})\Pi_{B}^{R}/2+\mu_{G}%
\Pi_{A}^{G}/2+(1-\mu_{G})\Pi_{B}^{G}/2, \label{eqwelfare}%
\end{equation}
where $\Pi_{A}^{X}\equiv\Pi_{A}^{X}(\mu_{R},\mu_{G})$ and $\Pi_{B}^{X}%
\equiv\Pi_{B}^{X}(\mu_{R},\mu_{G})$ are given by equations (\ref{pi_r_a}) and
(\ref{pi_r_b}). Since unemployed workers obtain zero utility, we can also
write the welfare function as
\begin{equation}
W(\mu_{R},\mu_{G})=L_{A}U\left(  \frac{\partial F}{\partial L_{A}}(L_{A}%
,L_{B})\right)  +L_{B}U\left(  \frac{\partial F}{\partial L_{B}}(L_{A}%
,L_{B})\right)  , \label{eqwelfare2}%
\end{equation}
where $L_{A}\equiv L_{A}(\mu_{R},\mu_{G})$ and $L_{B}\equiv L_{B}(\mu_{R}%
,\mu_{G})$ were introduced by (\ref{eqla}) and (\ref{eqlb}). The formulation
in (\ref{eqwelfare2}) is useful, because it shows that what matters for social
welfare is the effect of a policy on the society's effective labor supply.

We consider a first-best social optimum, that is, the social planner is able
to fully manage $\mu_{R} \in[0,1]$ and $\mu_{G} \in[0,1]$ and therefore the
social optimum $\mu^{S}=(\mu_{R}^{S}, \mu_{G}^{S})$ is defined as
\[
\mu^{S} = \mbox{argmax}_{\mu_{R} \in[0,1], \mu_{G} \in[0,1]} W(\mu_{R},
\mu_{G}).
\]
We obtain the following result:

\begin{proposition}
\label{propwelfare} If for all $x\in\lbrack0,(p+\kappa+\lambda)/2]:$
\begin{equation}
s^{\prime\prime}(x)>-\frac{4}{\lambda}\ s^{\prime}(x)
\end{equation}
then any social optima involves complete or partial segregation.
\end{proposition}

Thus a segregation policy is socially preferred, as long as $s(x)$, the
employment probability of having $x$ friends with the same education, is "not
too concave". This proposition can be intuitively understood as follows.
Suppose that there is no segregation, and $0<\mu_{G}<\mu_{R}<1$. In that case
the Reds obtain a higher employment probability in an $A$-occupation,
$s_{A}^{R}>s_{B}^{R}$, whereas the Greens have a higher employment rate as
$B$-educated workers, $s_{B}^{G}>s_{A}^{G}$. Now consider the effect on segregation,
wages and employment when a social planner can force a—sufficiently small—measure of Red individuals initially
choosing a $B$-occupation and respectively, the same measure of Green individuals initially
choosing an $A$-occupation, into \emph{switching their occupation choice}, such that
$\mu_{R}$ slightly increases, whereas $\mu_{G}$ slightly decreases with the same amount.
The result of this event is, first, that segregation increases: the gap
between $\mu_{R}$ and $\mu_{G}$ becomes larger. Second, the total fraction of
individuals that choose occupation $A$, $(\mu_{R}+\mu_{G})/2$, does not change. So
the total ratio of $A$-educated to $B$-educated workers does not change, and
therefore the ratio of equilibrium wages is not much affected either. Hence, the effect on
wage inequality is only marginal. Third, by switching occupations, those Red
workers can now benefit from a denser network, and have an employment
probability $s_{A}^{R}$ instead of $s_{B}^{R}$. The same is true for the Green
workers switching from $A$ to $B$. Thus, the combined payoff of those switching workers
increases, as they are all more likely to become employed.  
We also need to consider the externality on the employment rates of the workers not involved
in the occupation switch. In particular, the switch of occupations increases
the network effects of the other Red $A$-educated and Green $B$-educated workers,
whereas it decreases the network effects of Red $B$-educated and Green
$A$-educated workers. The restriction on the concavity of $s(x)$ ensures that the
switch of occupations puts on average a positive externality on the employment
probabilities of other workers. 
 We conclude that the occupational switch of the two equal measures of Red and Green workers hardly affects wage inequality, while it increases the labor
supply of both $A$ and $B$. Therefore, social welfare increases. This is analogous 
to the so-called "participation externality" rationale from equilibrium models with positive spillovers and 
strategic complementarities, where an agent's decision to participate (i.e., produce) in a market
depends on the number of active agents in that market, see, e.g., the "trading externalities"
search equilibrium by \citet{Diamond 1982} or the general characterization of such models by \citet{Cooper and John 1988}.

The general message of this result is that integration policies might also
have -- unintended -- detrimental effects.\footnote{Integration might obviously be desirable for many reasons not captured by our parsimonious model. For instance, both direct and indirect personal utility may benefit from exposure to diversity; or, particularly relevant in our context of labor markets with homophilous social networks and job referrals, more integration may be desirable for its benefits on intergenerational mobility, see, e.g., the discussions and references in \citet{Bolteetal} and \citet{Jackson 2022}.} Under our model's assumptions, integration might weaken the employment chances of individuals,
because network effects are weaker in mixed networks. In the case of
complete segregation, individuals are surrounded by similar individuals during
their education. Thus, it is easier for them to make many friends they
can rely on when searching on the job market. Consequently, employment
probabilities are high. On the other hand, if educations are mixed, then
individuals have more difficulties in creating useful job contacts, and
therefore their employment probabilities are lower. It is worth stressing
that the result that integration weakens network effects and decreases labor
market opportunities has empirical support in related literature on
segregation. For example, \citet{Currarini et al} find clear evidence that
larger (racial) minorities create more friendships, and \citet{Marsden 1987}
finds a similar pattern in his network of advice. Therefore, it is more
beneficial for a worker to choose an education in which she is only surrounded
by similar others, instead of an education in which racial groups are mixed,
let alone one in which she is a small minority. In a different but related
context, \citet{alesina1, alesina2} find that participation in social
activities is lower in racially mixed communities and so is the level of
trust. These and our results suggest that possible negative impacts of
integration on social network effects should also be taken into account -- and mitigated where integration is explicitly desired.

Our outcome on the first-best social optimum hinges for a large part on the
fact that the social planner is able to increase employment by directly increasing
segregation through forced occupational switches of workers, while hardly affecting wage inequality. In reality though, a
social planner may not have this amount of control. Perhaps more feasible
would be a policy in which the social planner enforces and stabilizes
integration, but where the exact allocation of workers to occupations is
still determined by individual incentives, thus envisaging a potential trade-off of segregation between
network benefits and inequality. This suggests a second-best analysis of social
welfare, in which -- as more detailed in the next subsection of the paper -- the social planner is not able to fully control
the measure of workers choosing one (or the other) education, $\mu_{R}$ and $\mu_{G}$, but could nevertheless \emph{stabilize a symmetric
equilibrium}, such that the same measure of workers in both social groups chooses a particular education, $\mu_{R}=\mu_{G}=\mu^{S}$. Such an analysis is however unfeasible without
further parameter specifications, hence we will perform that analysis
subsequent to calibrating the model for suitable parameters and functional forms.

\subsection{Second best social optimum}

\subsubsection{Numerical simulation}

\label{secsim}

We calibrate the parameters, in order to perform a small numerical simulation
of our model. The purpose of this simulation is to get a better feeling on the
mechanisms of the model, the restrictiveness of our assumptions, and the
magnitude of the wage gap that can be generated. This straightforward
simulation also allows us to get some key insights about a second-best welfare
policy. A detailed analysis would require an extension of the model and is
beyond the scope of this paper.

We first specify functional forms for $s(x)$, the employment probability as
function of the number of friends with the same education, $F(L_{A},L_{B})$,
the production function and thus the derived wage functions, and $U(x)$, the
utility function. Regarding the employment probability, we consider a function
that follows from a dynamic labor process, in which employed individuals
become unemployed at rate 1, and in which unemployed individuals become
employed at rate $c_{0}+c_{1}x$, where $c_{0}$ is the rate at which unemployed
workers directly obtain information on job vacancies, and $c_{1}$ measures the
strength of having friends.\footnote{All parameters considered here are defined on appropriate domains such that the 'composite parameters' that we actually calibrate are well defined in the ranges earlier stated in the paper. For instance, since $s_{0}=c_{0}/(1+c_{0})$ as introduced later, the $c_{0}$ introduced here is considered to be defined on the appropriate domain such that $s_{0}=s(0)$, with $0<s(0)<1$.} This leads to the following employment function:
\[
s(x)=\frac{c_{0}+c_{1}x}{1+c_{0}+c_{1}x}.
\]
Since we have defined $s_{0}=s(0)$ as the employment probability when only
direct search is used, it follows that $s_{0}=c_{0}/(1+c_{0})$ 

For the production function we assume the commonly used Cobb-Douglas function
with constant returns to scale,
\[
F(L_{A},L_{B})=\theta L_{A}^{\alpha}L_{B}^{1-\alpha}.
\]
For the utility function we consider a function with constant absolute risk
aversion, where $\rho$ is the coefficient of absolute risk aversion. That is
\[
U(x)=1-e^{-\rho x}.
\]

We calibrate the parameters $s_{0}, c_{1}(p+\kappa), c_{1} \lambda, p$ and
$\theta$, leaving $\alpha$ as a free parameter. First, we calibrate $s_{0}$,
$c_{1}(p+\kappa)$, and $c_{1}\lambda$ from three equations that are motivated
by the empirical evidence given in Section \ref{secsebi2}. This parameterization is
sufficient to perform the simulation, and it is thus not necessary to
separately specify $c_{1}$, $p$, $\kappa$ and $\lambda$. The first equation is
obtained by imposing the restriction that about 50 \% of the workers find their
job through friends, as suggested in Section \ref{secsebi2_jcn}. This restriction implies that
the direct job arrival rate $c_{0}$ should equal the indirect job arrival rate
through friends $c_{1}x$. The indirect job arrival rate differs, depending on
the choices of the individuals, but if we focus on the case complete
segregation, in which $\mu_{R}=1$ and $\mu_{G}=0$, then we can impose the
following restriction:
\[
c_{0}=c_{1}(p+\kappa+\lambda)/2.
\]

Next, we calibrate the amount of inbreeding homophily in the social group.
This amount typically differs depending on the group defining characteristic.
For example, analyzing data on Facebook participants at Texas A\&M,
\citet{MayerPuller} find that two students living in the same dorm are 13 times more likely to be friends than two random students, two black students 17 times more
likely, but two Asian students 5 times more likely, and two Hispanic students
twice as likely to be friends. In light of this evidence, we chose to keep the
amount of inbreeding homophily in the simulation modest, imposing
$\lambda=3(p+\kappa)$.

We next impose that the employment rate is 95 \% in case of complete segregation. Remark that this benchmark employment rate value is chosen for illustration purposes only, and that all outcomes are qualitatively robust for a range of values (we have tried 90 \% to 99 \%, as realistic in context). Below, at the end of this subsection, we make further comments on the sensitivity of this and all our other relevant parameters.
Given the above, we solve
\[
\frac{2c_{0}}{1+2c_{0}}=0.95,
\]
and this implies that
\[
s_{0}=\frac{c_{0}}{1+c_{0}}=\frac{19}{21}\approx.9048.
\]
and further that $c_{1}(p+\kappa)=4.75$ and $c_{1}\lambda=14.25$.

Let us consider now the productivity parameter $\theta$ and the coefficient of
absolute risk aversion $\rho$. The coefficient of absolute risk aversion has
been estimated between $6.6\times10^{-5}$ and $3.1\times10^{-4}$
(\citet{Gertner}, \citet{Metrick 1995}, \citet{coheneinav}). We set the risk
aversion at $1.0\times10^{-4}$, which means a coefficient of relative risk
aversion of 4 at a wealth level of \$ 40,000, or indifference at participating
in a lottery of getting \$ 100.00 or losing \$ 99.01 with equal probability.

The productivity parameter, $\theta$, is chosen such that average income equal
\$ 40,000 in the case of complete segregation, $(\mu_{R},\mu_{G})=(1,0)$, and
$\alpha=.5$. Since in that situation $w_{A}(1,0)=w_{B}(1,0)=\theta/2$, we have
$\theta=80,000$.

\begin{table}[ptb]
\caption{Chosen parameter values in the simulation and sensitivity w.r.t. $\hat{\alpha}$ and the maximum wage gap.}%
\label{tab:parameters}
\begin{center}%
\begin{tabular}
[c]{|l|ccc|}\hline\hline
Parameter & Value & Elasticity of $\hat{\alpha}$ & Elasticity of wage gap\\
&  & $\hat{\alpha}=.5904$ & $G(1,0)=.306$\\\hline
$s_{0}$ & .9048 & -1.71 & -9.47\\
$c_{1}(p+\kappa)$ & 4.75 & -.04 & -.23\\
$c_{1}\lambda$ & 14.25 & .08 & .46\\
$\rho$ & $1.0\times10^{-4}$ & .38 & 2.09\\
$\theta$ & 80,000 & .38 & 2.09\\\hline\hline
\end{tabular}
\end{center}
\end{table}

We now look at the dependence of payoffs, wages and employment on $\alpha$
with $s_{0}$, $c_{1}(p+\kappa)$, $c_{1}\lambda$, $\rho$ and $\theta$ as
summarized in Table \ref{tab:parameters}, and in which $\mu_{R}$ and $\mu_{G}$
are determined by equilibrium conditions (\ref{eqcondmu0})-(\ref{eqcondmu1}).
Given the result of Proposition \ref{propsegregation} that there is either a
complete equilibrium or a partial segregation equilibrium, in which one group specializes
in the good job, we concentrate our attention to the parameter space in which
$\alpha\in\lbrack1/2,1)$, $\mu_{R}=1$ and $\mu_{G}\in\lbrack0,1)$. Thus
occupation $A$ is \textquotedblleft good\textquotedblright, and group $R$
specializes in $A$.

We first show a plot of $\Delta\Pi^{G}(1,\mu_{G})$ as a function of $\mu_{G}$
for different values of $\alpha$. This function illustrates the payoff
evaluation that a Green individual makes when deciding on its occupation. If
$\Delta\Pi^{G}(1,\mu_{G})>(<)0$, then the Green individual prefers $A$ ($B$)
if she beliefs that all Reds choose $A$ and fraction $\mu_{G}$ of Greens
choose $A$. Clearly, in an equilibrium it should hold that either $\Delta
\Pi^{G}(1,0)<0$ or $\Delta\Pi^{G}(1,\mu_{G})=0$.%

\begin{figure}
[ptb]
\begin{center}
\includegraphics[
height=2.5in,            
width=3.88in         
]%
{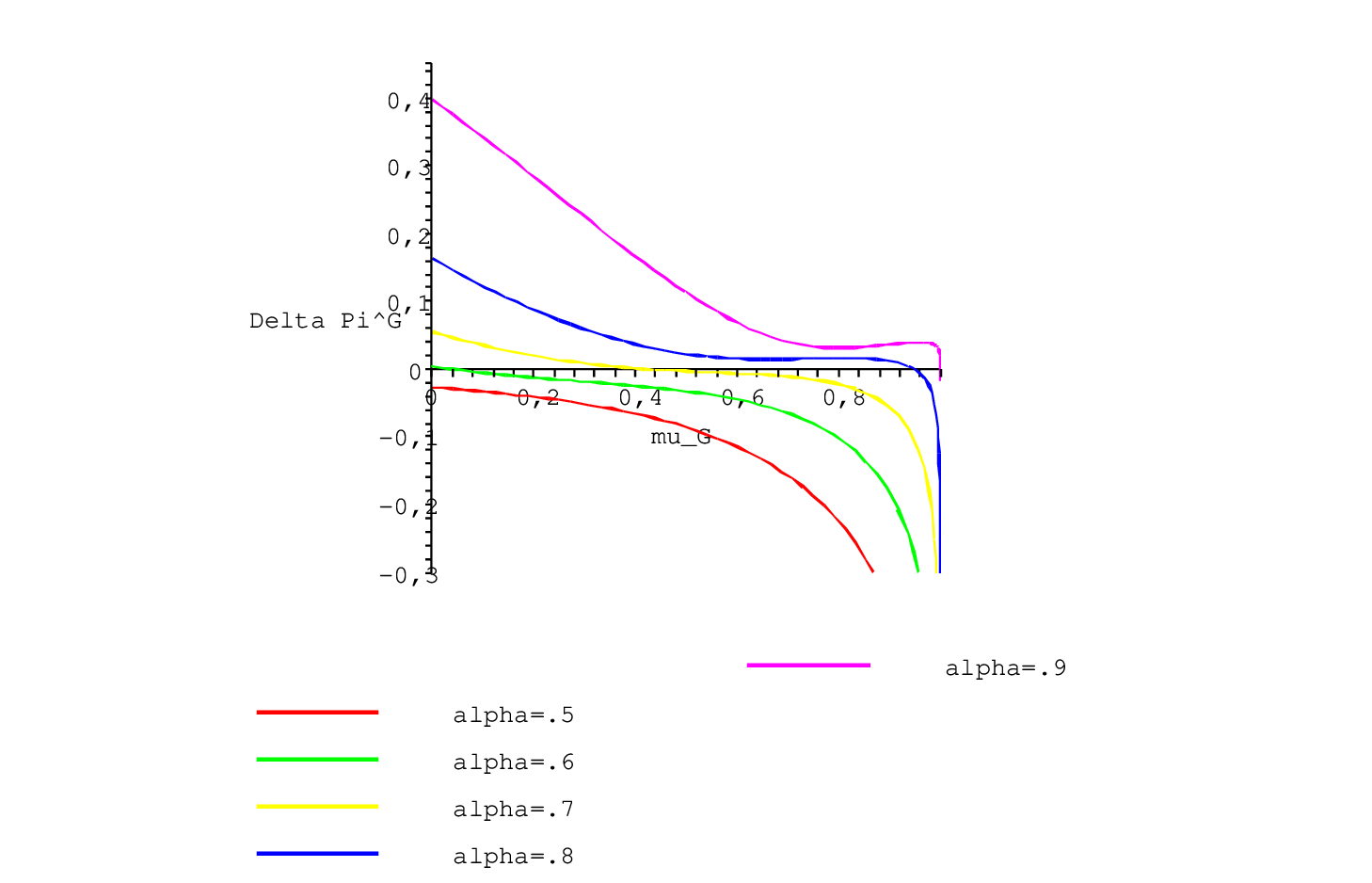}%
\caption{$\Delta\Pi^{G}(1,\mu_{G})$ as a function of $\mu_{G}$ for different
values of $\alpha$.}%
\label{figDPg1}%
\end{center}
\end{figure}

The plot is displayed in Figure \ref{figDPg1}. This plot nicely illustrates
the workings of the model. First, note that for $\alpha=.5$, $\Delta\Pi
^{G}(1,\mu_{G})$ is clearly negative, so given that the Reds choose $A,$ the
Greens prefer $B$ and complete segregation is an equilibrium. However,
$\Delta\Pi^{G}(1,\mu_{G})$ increases with $\alpha$, such that for
$\alpha>.5904\equiv\hat{\alpha}$, we have that $\Delta\Pi^{G}(1,0)>0$ and
complete segregation is not an equilibrium anymore. In that case, there is a
unique partial segregation equilibrium.\footnote{$\Delta\Pi^{G}(1,\mu_{G})$ is not monotonically
decreasing for very large $\alpha$, which implies that Assumption
\ref{assemployment} is violated. Nonetheless, there is still a unique
equilibrium for all values of $\alpha$.}

If $\alpha<.5904$ we have complete segregation as an equilibrium. In that case the equilibrium employment rates and the wages are
characterized by the Proposition \ref{prop4b}.\ Given our parameterization, employment
rates are given by:
\[
s_{A}^{R}=s_{B}^{G}=.95\mbox{ and }s_{B}^{R}=s_{A}^{G}=.9223.
\]
Wages have a particular simple form in the case of complete segregation, being
$w_{A}(1,0)=\theta\alpha$ and $w_{B}(1,0)=\theta(1-\alpha)$. Therefore, if we
define the wage gap as $G(\mu_{R},\mu_{G})=1-w_{B}(\mu_{R},\mu_{G})/w_{A}%
(\mu_{R},\mu_{G})$, then the wage gap under complete segregation is
$G(1,0)=2-1/\alpha$. Note that at $\alpha=\hat{\alpha}=.5904$, we have
\[
w_{A}(1,0)=47,233\mbox{ and }w_{B}(1,0)=32,767
\]
and the wage gap is thus $G(1,0)=.306$. Hence, a small employment gap of .9223
versus .95 is only compensated by a wage gap of 30 \%!

It is worth elaborating on this potentially large wage gap. In equilibrium,
group $R$ is completely specialized in education $A$. Therefore the wage and
unemployment gap are determined by the tradeoff that workers from group $G$
are making. Choosing education $A$ gives Greens a higher wage than
education $B$, but in education $B$ there would be few Green colleagues, and
therefore fewer contacts. Hence, choosing $A$ would result in a lower
employment rate for Green workers. Remark that this
unemployment gap may be quite small compared to the wage gap. In particular,
in our simulation, at $\alpha=\hat
{\alpha}\equiv.5904$, the wage gap of 30 \% is compensated by an employment
gap of about 3 \%. The reason for this tenfold magnification is risk aversion
of individuals. Individuals try to avoid the (small) risk of unemployment, in
which they have a payoff equal to 0, being willing to accept even major
losses in income to accomplish that.\footnote{The risk aversion effect, and
thus the wage gap, may be smaller if unemployment is only temporary, and
individuals only care about permanent income, or if agents get unemployment
benefits/ social support. On the other hand, from prospect theory it is known
that individual agents tend to emphasize small probabilities
(\citet{KahnemanTversky}), and thus the small probability of becoming
unemployed may get excessive weight in the education decision.}

We would like to know whether an even larger wage gap can be sustained in a
partial segregation equilibrium when $\alpha>\hat{\alpha}=.5904$. We therefore
plot the equilibrium wages, $w_{A}(1,\mu^{\ast})$ and $w_{B}(1,\mu^{\ast})$,
and equilibrium employments, $s_{A}^{R}(1,\mu^{\ast})$, $s_{B}^{R}(1,\mu
^{\ast})$, $s_{A}^{G}(1,\mu^{\ast})$ and $s_{B}^{G}(1,\mu^{\ast})$, as
function of $\alpha$. Remember that the equilibrium $\mu^{\ast}$ equals zero
when $\alpha\leq\hat{\alpha}$, and solves $\Delta\Pi^{G}(1,\mu^{\ast})=0$ when
$\alpha>\hat{\alpha}$. These plots are shown in Figures~\ref{figwages1}%
~and~\ref{figemployment1}.%

\begin{figure}
[ptb]
\begin{center}
\includegraphics[
height=2.5in,                 
width=2.87389in                 
]%
{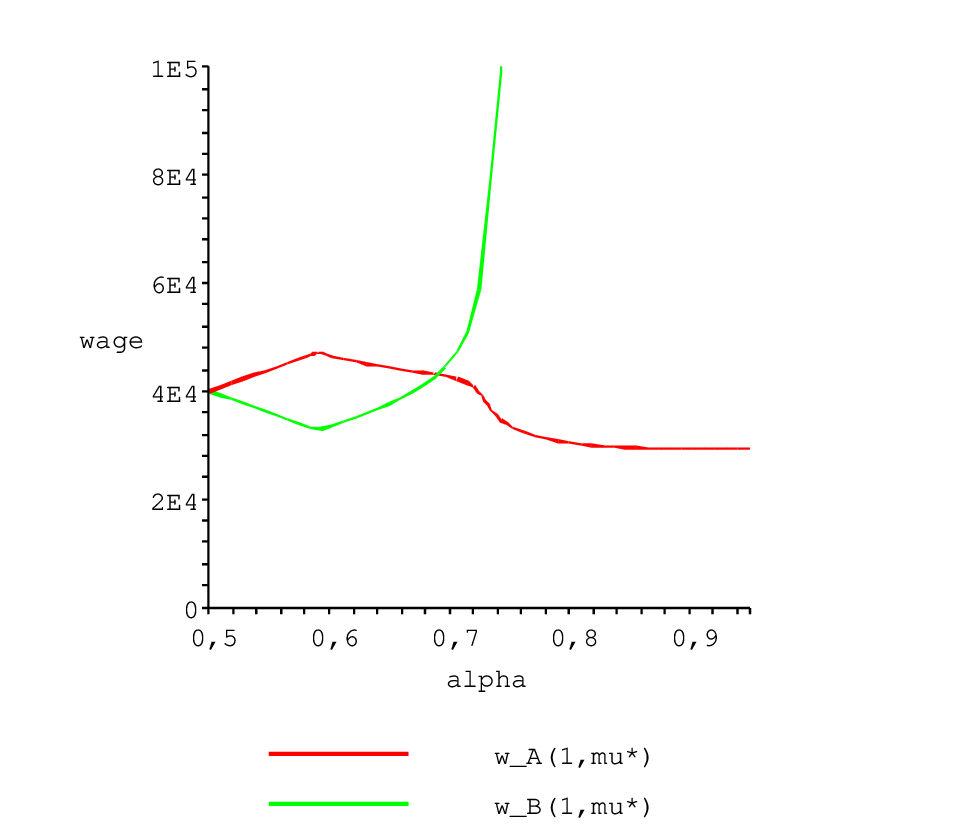}%
\caption{Equilibrium wages as function of $\alpha$.}%
\label{figwages1}%
\end{center}
\end{figure}
%

\begin{figure}
[ptb]
\begin{center}
\includegraphics[
height=2.5in,                    
width=2.89836in                           
]%
{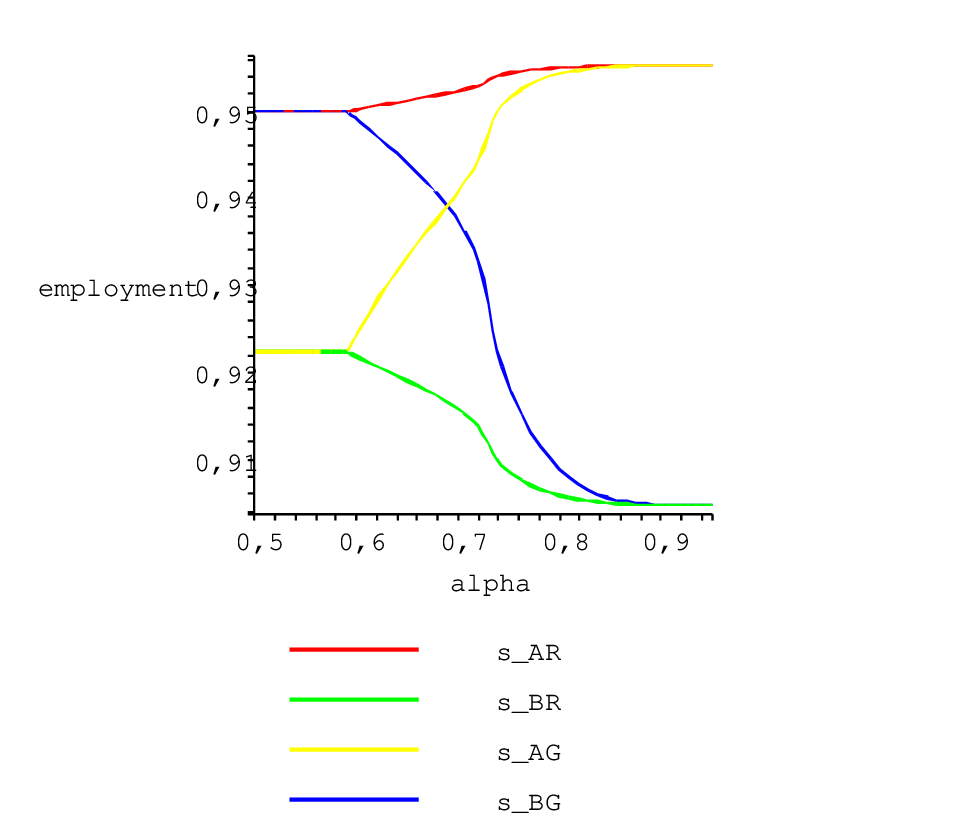}%
\caption{Equilibrium employment rates as function of $\alpha$.}%
\label{figemployment1}%
\end{center}
\end{figure}

The figures above embody well the qualitative implications of Propositions~\ref{prop4b}~and~\ref{prop6}, for the chosen parameter values/ ranges.
Moreover, for the chosen parameters we also observe that the wage gap
$G(1,\mu^{\ast})$ is maximized at $\alpha=\hat{\alpha}$. When $\alpha$ becomes
larger than $\hat{\alpha}$, the wage of $A$ declines and the wage of $B$
increases until the wage gap is reversed.

We next look at the sensitivity of $\hat{\alpha}$ with respect to the
parameter choices, as we saw that at $\alpha=\hat{\alpha}$ the wage gap is
maximized. We do this by computing the elasticities of $\hat{\alpha}$ and of
the implied wage gap $G(1,0)$ at the chosen parameters. That is, we look at
the percentage increase of $\hat{\alpha}$ and the maximum wage gap change when
a parameter increases by 1 \% . The elasticities are shown in columns~2~and~3
of Table~\ref{tab:parameters}. We note that $\hat{\alpha}$ and the implied
maximum wage gap are most sensitive to $\rho\theta$, the coefficient of
relative risk aversion. A 1 \% increase in this coefficient leads to a 2 \%
increase in the maximum wage gap. On the other hand, our calibration seems
least sensitive to the network parameters $c_{1}(p+\kappa)$ and $c_{1}\lambda
$. The maximum wage gap seems to be close to linear with respect to $1-s_{0}$,
the unemployment rate if a worker only consider direct search techniques. That
is, if we chose $s_{0}=.95$ instead of $s_{0}=.90$, it would roughly halve the
maximum wage gap.

\subsubsection{Implications for the second-best welfare outcome}

We now consider explicitly the analysis of a second-best optimum. Namely, as also briefly stated earlier in the paper, we suppose that
the government (social planner) does not have the institutions to completely
control $\mu_{R}$ and $\mu_{G}$, but that it is still able to stabilize a symmetric
equilibrium, such that $\mu_{R}=\mu_{G}=\mu^{S}$.\footnote{In the proof of
Lemma~\ref{lemma1} we showed that there exists a symmetric equilibrium, but that
it is unstable; that is, after a small deviation from the equilibrium,
individual incentives drive education choices to segregation.} Should
the government do this? In case the government stabilizes integration, we
still impose the equilibrium condition, which is in this case symmetric.
Therefore
\[
\Pi_{A}^{R}(\mu^{S},\mu^{S})=\Pi_{B}^{R}(\mu^{S},\mu^{S})=\Pi_{A}^{G}(\mu
^{S},\mu^{S})=\Pi_{B}^{G}(\mu^{S},\mu^{S}).
\]
Hence, in the symmetric case there is complete equality. On the other hand, in
the case of segregation, we consider the equilibrium allocation $(\mu_{R}%
,\mu_{G})=(1,\mu^{\ast})$, such that Reds obtain a higher payoff than Greens.
Therefore, we might face a tradeoff when assessing an integration policy. It
enforces equality, but it might decrease employment.

To this purpose we plot the increase in social welfare from such an
integration policy, $I=W(\mu^{S},\mu^{S})/W(1,\mu^{\ast})-1$, as function of
$\alpha$. Figure~\ref{figsecondbest1} shows this plot.%

\begin{figure}
[ptbh]
\begin{center}
\includegraphics[
height=2in,                     
width=2in                          
]%
{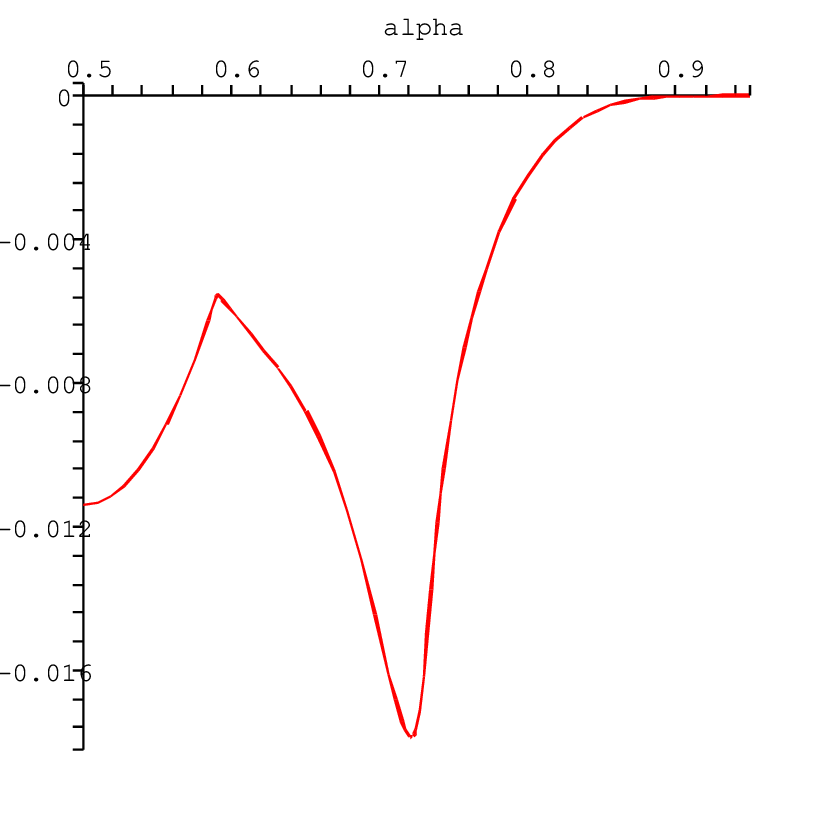}%
\caption{Percentage increase in welfare of a policy that enforces perfect
integration.}%
\label{figsecondbest1}%
\end{center}
\end{figure}

We observe that $I$ is negative for all values of $\alpha$. So for the chosen
parameters the integration policy is never preferred. In this world, people are better off segregated.

Our results are clear: with a utilitarian welfare function, a second best
policy involves a \textquotedblleft laissez-faire\textquotedblright\ policy,
such that society becomes segregated. The intuition behind this result is
twofold. First, in the case of partial segregation the equilibrium is
determined by the Green workers. They trade off a benefit in wage against a
loss in employment. Their individual incentives therefore already put a limit
on the amount of wage inequality that can be sustained in equilibrium. Second,
an integration policy would lead to lower employment rates. Or, with
risk-averse individuals, the society is willing to tolerate some inequality in exchange for higher employment rates.

We finally remark that an integration policy is only beneficial when society
has \emph{additional} distributional concerns that are not captured by the
concavity of the individual utility function. For example, consider the case
of a maximin social welfare function: $W_{\min}=\min_{i}\Pi_{i}$. In the
integrated case, $\mu_{R}=\mu_{G}=\mu^{S}$, everyone obtains the same payoff,
whereas in the segregated case workers from group $G$ are worse off.
Therefore, $W_{\min}(1,\mu^{\ast})=\Pi_{B}^{G}(1,\mu^{\ast})$ and $W_{\min
}(\mu^{S},\mu^{S})=\Pi_{B}^{G}(\mu^{S},\mu^{S})$. We show a comparison of
these two payoffs, $\Pi_{B}^{G}(\mu^{S},\mu^{S})/\Pi_{B}^{G}(1,\mu^{\ast})-1$,
in Figure~\ref{figmaximin1}.%

\begin{figure}
[h]
\begin{center}
\includegraphics[
height=2in,                    
width=2.17164in                     
]%
{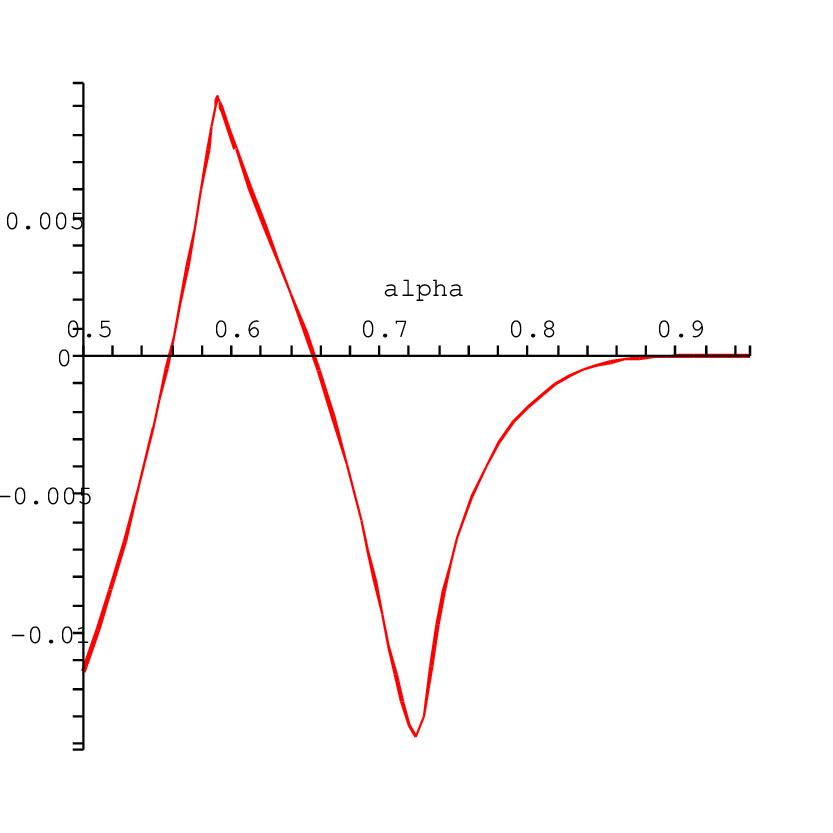}%
\caption{Percentage increase in Green workers' payoffs of a policy enforcing
perfect integration.}%
\label{figmaximin1}%
\end{center}
\end{figure}

Observe that Greens would benefit from integration for values
of $\alpha$ around $\hat{\alpha}$, where the wage gap is particularly large.
In such case, strong distributional concerns would justify
integration\footnote{\citet{Graham et al} provide a conceptual framework to test
for equity-efficiency tradeoffs, focusing on \textit{local segregation
inequality effects}. One can rule segregation-increasing efficiency gains
unacceptable if they increase inequality across groups.}.

\section{Conclusion}

\label{secsebi6}

We have proposed a social interaction model with jobs obtained through a
stochastic network of contacts, after individual career decisions had been
made endogenously. Even with a tiny amount of within-social group homophily,
a partial occupational segregation equilibrium, in which one
group fully specializes and the other mixes over two
career tracks, can be sustainable -- with between-group wage and employment gaps also in line with observed gender or racial labor market disparities.

We have also analysed the implications of our model from a social planner's perspective.
In the first best social welfare optimum, segregation is
socially preferred. Subject to proper calibration of our model
parameters, a second best social welfare analysis supports a laissez-faire
policy where society also becomes segregated, shaped by individual
incentives. Both conclusions are valid in light of reasonable concavity
features of the individual utility function. Our results imply that an integration policy could only
be justified under distributional concerns beyond typical individual utilities; or under more complex individual preferences than those analysed in our model, for instance with direct utility derived from exposure to more diversity; and therefore that such justifications should take center stage in social integration debates. 

While our model can relate empirical patterns of educational and occupational
segregation to wage and employment gaps between gender or racial groups,
other factors have also been documented to play a significant role.
This framework should thus be seen as \textit{complement} to a number of
alternatives, including the classical theories of taste discrimination or
rational bias by employers. In future empirical research it would be pertinent to assess
the relative strength of our mechanism vis-\`{a}-vis those other channels in
explaining labor market disparities.

Our model easily allows for interesting extensions. One such avenue for future
research is to extend with an analysis of minority versus majority groups, by
modeling the interaction between social groups of unequal
sizes.\footnote{A recent paper by \citet{Okafor 2020} models minority
vs. majority labor outcome differentials, also in a social network setup with
homophily;\ however, unlike our straightforward model extension envisaged
above, \citet{Okafor 2020} hinges crucially on the inequality in the size of
the social groups considered.} Another avenue is to consider heterogeneity in
productivity. The latter would allow to analyze the mismatch of workers to
firms due to network effects. Finally, one could further explore 
whether changes in our parsimonious economy fundamentals, e.g., in the number of (initial or newly entering) occupations and 
in the degree of complementarity between these occupations, could alter the equilibrium towards less segregation relative to the first-best.

\appendix

\section{Proofs}

The proof of Proposition~\ref{propsegregation} uses the following lemma:

\begin{lemma}
Suppose Assumptions \ref{asswage} and \ref{assemployment} hold. A stable equilibrium $(\mu_{R}^{\ast},\mu_{G}^{\ast})$ in which $0<\mu
_{R}^{\ast}<1$ and $0<\mu_{G}^{\ast}<1$ does not exist. \label{lemma1}
\end{lemma}

\begin{proof}
Suppose $(\mu_{R}^{\ast},\mu_{G}^{\ast})$ is a stable equilibrium, and
$\mu_{R}^{\ast}\in(0,1)$ and $\mu_{G}^{\ast}\in(0,1)$. By
condition~(\ref{eqcondmu})
\begin{equation}
\Pi_{A}^{R}(\mu_{R}^{\ast},\mu_{G}^{\ast})=\Pi_{B}^{R}(\mu_{R}^{\ast},\mu
_{G}^{\ast})\mbox{ and }\Pi_{A}^{G}(\mu_{R}^{\ast},\mu_{G}^{\ast})=\Pi_{B}%
^{G}(\mu_{R}^{\ast},\mu_{G}^{\ast}) \label{eqcondrg}%
\end{equation}
Substituting (\ref{pi_r_a})-(\ref{pi_r_b}) into (\ref{eqcondrg}) and
rewriting, these equations become
\begin{equation}
\frac{U(w_{A}(\mu_{R}^{\ast},\mu_{G}^{\ast}))}{U(w_{B}((\mu_{R}^{\ast},\mu
_{G}^{\ast}))}=\frac{s_{B}^{R}(\mu_{R}^{\ast},\mu_{G}^{\ast})}{s_{A}^{R}%
(\mu_{R}^{\ast},\mu_{G}^{\ast})}=\frac{s_{B}^{G}(\mu_{R}^{\ast},\mu_{G}^{\ast
})}{s_{A}^{G}(\mu_{R}^{\ast},\mu_{G}^{\ast})}. \label{eqcondrg2}%
\end{equation}
Since $\lambda>0$, $\mu_{R}^{\ast}>\mu_{G}^{\ast}$ implies $s_{A}^{R}%
>s_{A}^{G}$ and $s_{B}^{R}<s_{B}^{G}$. But this means that if $\mu_{R}^{\ast
}>\mu_{G}^{\ast}$, then
\[
\frac{s_{B}^{R}(\mu_{R}^{\ast},\mu_{G}^{\ast})}{s_{A}^{R}(\mu_{R}^{\ast}%
,\mu_{G}^{\ast})}<\frac{s_{B}^{G}(\mu_{R}^{\ast},\mu_{G}^{\ast})}{s_{A}%
^{G}(\mu_{R}^{\ast},\mu_{G}^{\ast})}.
\]
which contradicts (\ref{eqcondrg2}). The same reasoning holds for $\mu
_{R}^{\ast}<\mu_{G}^{\ast}$. Hence, it must be that $\mu_{R}^{\ast}=\mu
_{G}^{\ast}$.
However $(\mu_{R}^{*}, \mu_{G}^{*})$ with $\mu_{R}^{*}=\mu_{G}^{*}$ cannot be
a stable equilibrium. To see this, suppose that $(\mu^{*},\mu^{*})$ with
$\mu^{*}\in(0,1)$ is a stable equilibrium. Hence $\Pi_{A}^{X}(\mu^{*},\mu
^{*})=\Pi_{B}^{X}(\mu^{*},\mu^{*})$ for $X\in\{R,G\}$ and $\frac
{\partial\Delta\Pi^{X}}{\partial\mu_{X}}< 0$ at $\mu_{R}=\mu_{G}=\mu^{*}$, and
det$(G(\mu^{*},\mu^{*})> 0$, where $G(\mu)=D\Delta\Pi(\mu)$ is the Jacobian of
$\Delta\Pi\equiv(\Delta\Pi^{R}, \Delta\Pi^{G})$ with respect to $\mu\equiv
(\mu_{R}, \mu_{G})$.
Since $\lambda>0$, it must be that
\begin{equation}
\frac{\partial s_{A}^{X}}{\partial\mu_{X}}>\frac{\partial s_{A}^{X}}%
{\partial\mu_{Y}}>0 \label{eqprop2dsadmu}%
\end{equation}
and
\begin{equation}
\frac{\partial s_{B}^{X}}{\partial\mu_{X}}<\frac{\partial s_{B}^{X}}%
{\partial\mu_{Y}}<0 \label{eqprop2dsbdmu}%
\end{equation}
for $X,Y\in\{R,G\}$ and $Y\neq X$. Furthermore, if $\mu_{R}=\mu_{G}=\mu^{\ast
}$, then $s_{A}^{X}=s_{A}^{Y}$, $\frac{\partial L_{A}}{\partial\mu_{X}}%
=\frac{\partial L_{A}}{\partial\mu_{Y}}$, $\frac{\partial L_{B}}{\partial
\mu_{X}}=\frac{\partial L_{B}}{\partial\mu_{Y}}$, and therefore,
\begin{equation}
\frac{\partial w_{A}}{\partial\mu_{X}}=\frac{\partial w_{A}}{\partial\mu_{Y}}
\label{eqprop2dwadmu}%
\end{equation}
and
\begin{equation}
\frac{\partial w_{B}}{\partial\mu_{X}}=\frac{\partial w_{B}}{\partial\mu_{Y}}.
\label{eqprop2dwbdmu}%
\end{equation}
By using (\ref{eqprop2dsadmu})-(\ref{eqprop2dwbdmu}) from above, together with a substitution of (\ref{pi_r_a})-(\ref{pi_r_b}) in
Assumption~\ref{assemployment} from the main body of the paper, and a somewhat tedious but straightforward algebra, it follows that, at $\mu_{R}=\mu_{G}=\mu
^{\ast}$,
\[
\frac{\partial\Delta\Pi^{X}}{\partial\mu_{Y}}<\frac{\partial\Delta\Pi^{X}%
}{\partial\mu_{X}}<0.
\]
for $X,Y\in\{R,G\}$, $X\neq Y$.
But given the latest established inequalities, we can immediately see from the expression of the determinant of the Jacobian $G(\mu)=D\Delta\Pi(\mu)$ introduced above,  that det$(G(\mu^{\ast},\mu^{\ast}))<0$, which contradicts the stability assumption.
\end{proof}

\bigskip

\noindent PROPOSITION \ref{propsegregation} \begin{proof}
(i) If (\ref{prop3cond1}) holds,
then
\[
\Pi_{A}^{R}(1,0)>\Pi_{B}^{R}(1,0)\mbox{ and }\Pi_{A}^{G}(1,0)<\Pi_{B}%
^{G}(1,0).
\]
Hence, $(\mu_{R},\mu_{G})=(1,0)$ is clearly a stable equilibrium. The same is
true for $(\mu_{R},\mu_{G})=(0,1)$. Lemma~\ref{lemma1} and
Assumption~\ref{assemployment} ensure that these are the only two equilibria.
(ii) If (\ref{prop3cond2}) is true, then
\begin{equation}
\Pi_{A}^{G}(1,0)>\Pi_{B}^{G}(1,0). \label{proof5deviate}%
\end{equation}
Furthermore, from Assumption~\ref{asswage} we know that $\frac{\partial
\Delta\Pi^{G}(1,\mu_{G}) }{\partial\mu_{G}} < 0$ for all $\mu_{G} \in[0,1]$.
It follows from Assumption~\ref{asswage}, equation (\ref{proof5deviate}) and
continuity of $F$, $U$ and $s$, that there must be a unique $\mu^{*}$, such
that
\[
\Pi_{A}^{G}(1,\mu^{*})=\Pi_{B}^{G}(1,\mu^{*}).
\]
Moreover, $s_{A}^{R}(1,\mu^{*})>s_{A}^{G}(1,\mu^{*})$ and $s_{B}^{G}(1,\mu
^{*})>s_{B}^{R}(1,\mu^{*})$, so we have at $(\mu_{R},\mu_{G})=(1,\mu^{*})$
\begin{equation}
\Pi_{A}^{R}>\Pi_{B}^{G}=\Pi_{A}^{G}>\Pi_{B}^{R}. \label{eqpicomparisons}%
\end{equation}
It is therefore clear that $(\mu_{R},\mu_{G})=(1,\mu^{*})$ is a stable
equilibrium. The same is true for $(\mu_{R},\mu_{G})=(\mu^{*},1)$.
To show that there is no other equilibrium, note that by (\ref{prop3cond2}),
$\Pi_{A}^{R}(1,0)>\Pi_{B}^{R}(1,0).$ Assumption~\ref{assemployment} then
implies that $\Pi_{A}^{R}(\mu,0)>\Pi_{B}^{R}(\mu,0)$ for all $\mu\in[0,1]$.
Hence, $(\mu,0)$ and, similarly, $(0,\mu)$ cannot be an equilibrium. By
Lemma~\ref{lemma1} we also know that there is no mixed equilibrium.
\end{proof}

\bigskip


\noindent PROPOSITION \ref{prop4b} \begin{proof}
The equations follow almost directly. We
have
\[
s_{A}^{R}(1,0)=s_{B}^{G}(1,0)=s_{H}>s_{L}=s_{B}^{R}(1,0)=s_{A}^{G}(1,0).
\]
Further, by assumption $w_{A}\geq w_{B}$ at $(\mu_{R},\mu_{G})=(1,0)$.
Finally, at $(\mu_{R},\mu_{G})=(1,0)$
\[
U(w_{A})s_{A}^{R}\geq U(w_{B})s_{B}^{G}\geq U(w_{A})s_{A}^{G}\geq
U(w_{B})s_{B}^{R},
\]
and this is equivalent to (\ref{prop3comppayoff}).
\end{proof}

\bigskip


\noindent PROPOSITION \ref{prop6} \begin{proof}
Consider the stable equilibrium at
$(1,\mu^{\ast})$. Since it is an equilibrium we know that
\[
\Pi_{A}^{G}(1,\mu^{\ast})=\Pi_{B}^{G}(1,\mu^{\ast}).
\]
In the proof of Proposition~\ref{propsegregation}, equation
(\ref{eqpicomparisons}), we already demonstrated the inequality
(\ref{prop5comppayoff}) Further, by Assumption~\ref{assemployment} we know
that $\Delta\Pi^{G}(1,\mu_{G})$ is strictly monotonically decreasing in
$\mu_{G}$.
(i) If $\hat{\mu}<\frac{\lambda}{2(p+\kappa+\lambda)}$, then $s_{A}^{G}%
(1,\hat{\mu})<s_{B}^{G}(1,\hat{\mu})$. As $w_{A}(1,\hat{\mu}) =w_{B}%
(1,\hat{\mu})$ it must be that
\[
\Pi_{A}^{G}(1,\hat{\mu})<\Pi_{B}^{G}(1,\hat{\mu}).
\]
But then it also must be that $\mu^{\ast}<\hat{\mu}$. As we consider a partial segregation equilibrium, we know that $\mu^{\ast}>0$. Hence, $0<\mu^{\ast}<\hat{\mu}$ and
$w_{A}(1,\hat{\mu^{\ast}})>w_{B}(1,\hat{\mu^{\ast}})$, as $w_{A}(\mu_{R}%
,\mu_{G})$ is a decreasing function, whereas $w_{B}(\mu_{R},\mu_{G})$ is increasing. But then it also follows that $s_{A}^{G}(1,\hat{\mu^{\ast}})<s_{B}^{G}(1,\hat{\mu^{\ast}})$, 
since $\Pi_{A}^{G}(1,\hat{\mu^{\ast}})=\Pi_{B}^{G}(1,\hat{\mu^{\ast}})$ in equilibrium, as stated above.

(ii) If $\hat{\mu}>\frac{\lambda}{2(p+\kappa+\lambda)}$, then $s_{A}%
^{G}(1,\hat{\mu})>s_{B}^{G}(1,\hat{\mu})$ and $\Pi_{A}^{G}(1,\hat{\mu}%
)<\Pi_{B}^{G}(1,\hat{\mu}).$ But then $\mu^{\ast}>\hat{\mu}$. By
Assumption~\ref{asswage} we know that $\mu^{\ast}<1$. Hence, $\hat{\mu}%
<\mu^{\ast}<1$, and therefore $w_{A}(1,\hat{\mu^{\ast}})<w_{B}(1,\hat
{\mu^{\ast}})$. Analous to the above, this also further implies now that $s_{A}^{G}(1,\hat{\mu^{\ast}})>s_{B}^{G}(1,\hat{\mu^{\ast}})$,
since $\Pi_{A}^{G}(1,\hat{\mu^{\ast}})=\Pi_{B}^{G}(1,\hat{\mu^{\ast}})$ in equilibrium.
\end{proof}

\bigskip

We next continue with the proof of Proposition~\ref{propwelfare}. This proof
uses the following lemma:

\begin{lemma}
Suppose that for all $x\in[0, (p+\kappa+\lambda)/2]$
\begin{equation}
s^{\prime\prime}(x) > -\frac{4}{\lambda}\ s^{\prime}(x). \label{eqlemma2}%
\end{equation}

\begin{itemize}
\item[(i)] If $\mu_{X}>\mu_{Y}$ for $X,Y \in\{R,G\}$, then
\begin{equation}
\frac{\partial L_{A}}{\partial\mu_{X}}(\mu_{R}, \mu_{G})>\frac{\partial L_{A}%
}{\partial\mu_{Y}}(\mu_{R}, \mu_{G})>0, \label{eqlemmahyp1}%
\end{equation}
and
\begin{equation}
\frac{\partial L_{B}}{\partial\mu_{Y}}(\mu_{R}, \mu_{G})<\frac{\partial L_{B}%
}{\partial\mu_{X}}(\mu_{R}, \mu_{G})<0. \label{eqlemmahyp2}%
\end{equation}

\item[(ii)] If $\mu_{R}=\mu_{G}=\mu$, then
\begin{equation}
\frac{\partial^{2} L_{A}}{(\partial\mu_{X})^{2}}(\mu, \mu)>\frac{\partial^{2}
L_{A}}{\partial\mu_{X} \partial\mu_{Y}}(\mu,\mu), \label{eqlemmahyp3}%
\end{equation}
and
\begin{equation}
\frac{\partial^{2} L_{B}}{(\partial\mu_{X})^{2}}(\mu, \mu)>\frac{\partial^{2}
L_{B}}{\partial\mu_{X} \partial\mu_{Y}}(\mu,\mu). \label{eqlemmahyp4}%
\end{equation}

\end{itemize}

\label{lemma2}
\end{lemma}

\begin{proof}
(i) It is easy to derive that for $X\in\{R,G\}$:
\begin{equation}
\frac{\partial L_{A}}{\partial\mu_{X}} = \frac{1}{2}\left(  s_{A}^{X} +
\mu_{R} \frac{\partial s_{A}^{R}}{\partial\mu_{X}} + \mu_{G} \frac{\partial
s_{A}^{G}}{\partial\mu_{X}} \right)  > 0 \label{eqdladmx}%
\end{equation}
\begin{equation}
\frac{\partial L_{B}}{\partial\mu_{X}} = \frac{1}{2}\left(  -s_{B}^{X} +
(1-\mu_{R}) \frac{\partial s_{B}^{R}}{\partial\mu_{X}} + (1-\mu_{G})
\frac{\partial s_{B}^{G}}{\partial\mu_{X}} \right)  < 0 \label{eqdlbdmx}%
\end{equation}
at $(\mu_{R}, \mu_{G})$. From (\ref{eqdladmx}) and (\ref{eqdlbdmx}), we find
that for all $X,Y\in\{R,G\}: \partial L_{A}/ \partial\mu_{X} > \partial L_{A}/
\partial\mu_{Y}$ is equivalent to
\begin{equation}
s_{A}^{X} + \mu_{X} \left(  \frac{\partial s_{A}^{X}}{\partial\mu_{X}}%
-\frac{\partial s_{A}^{X}}{\partial\mu_{Y}}\right)  > s_{A}^{Y} + \mu_{Y}
\left(  \frac{\partial s_{A}^{Y}}{\partial\mu_{Y}}-\frac{\partial s_{A}^{Y}%
}{\partial\mu_{X}}\right)  . \label{eq1lemma2}%
\end{equation}
With the definition of $s_{A}^{X}$ in (\ref{eqsax}) we can write out
\begin{equation}
s_{A}^{X} + \mu_{X} \left(  \frac{\partial s_{A}^{X}}{\partial\mu_{X}}%
-\frac{\partial s_{A}^{X}}{\partial\mu_{Y}}\right)  = s\left(  (p+\kappa
)\bar{\mu}+\lambda\mu_{X}/2\right)  +\frac{\mu_{X} \lambda}{2} s^{\prime
}\left(  (p+\kappa)\bar{\mu}+\lambda\mu_{X}/2\right)  \label{eq2lemma2}%
\end{equation}
when $X\neq Y$. Therefore $\mu_{X}>\mu_{Y}$ is equivalent to (\ref{eq1lemma2}%
), whenever (\ref{eq2lemma2}) is strictly monotone increasing with $\mu_{X}$,
where we can treat $\bar{\mu}=(\mu_{X}+\mu_{Y})/2$ as a constant. It is easy
to check that this is indeed the case under condition (\ref{eqlemma2}). We
conclude that hypothesis (\ref{eqlemmahyp1}) holds whenever $\mu_{X}>\mu_{Y}$.
With a similar derivation one can show that condition (\ref{eqlemma2}) implies
(\ref{eqlemmahyp2}) as well.
(ii) The second derivatives of $L_{A}$ and $L_{B}$ with respect to $\mu_{X}$
and $\mu_{Y}$ are
\begin{equation}
\frac{\partial^{2}L_{A}}{\partial\mu_{X}\partial\mu_{Y}}=\frac{1}{2}\left(
\frac{\partial s_{A}^{X}}{\partial\mu_{Y}}+\frac{\partial s_{A}^{Y}}%
{\partial\mu_{X}}+\mu_{R}\frac{\partial^{2}s_{A}^{R}}{\partial\mu_{X}%
\partial\mu_{Y}}+\mu_{G}\frac{\partial^{2}s_{A}^{G}}{\partial\mu_{X}%
\partial\mu_{Y}}\right)  \label{eqd2ladmxdmy}%
\end{equation}%
\begin{equation}
\frac{\partial^{2}L_{B}}{\partial\mu_{X}\partial\mu_{Y}}=\frac{1}{2}\left(
-\frac{\partial s_{B}^{X}}{\partial\mu_{Y}}-\frac{\partial s_{B}^{Y}}%
{\partial\mu_{X}}+(1-\mu_{R})\frac{\partial^{2}s_{B}^{R}}{\partial\mu
_{X}\partial\mu_{Y}}+(1-\mu_{G})\frac{\partial^{2}s_{B}^{G}}{\partial\mu
_{X}\partial\mu_{Y}}\right)  . \label{eqd2lbdmxdmy}%
\end{equation}
Taking the second derivatives of $s_{A}^{X}$, evaluating at $\mu_{R}=\mu
_{G}=\mu$ and reordering, we get that (\ref{eqlemmahyp3}) is equivalent to
\begin{equation}
s^{\prime\prime}((p+\kappa+\lambda)\mu/2)<-\frac{4}{\lambda\mu}s^{\prime
}((p+\kappa+\lambda)\mu/2). \label{eqlemmahyp3re}%
\end{equation}
Inequality (\ref{eqlemmahyp3re}) clearly holds if condition (\ref{eqlemma2})
holds, which proves (\ref{eqlemmahyp3}). In a similar fashion, (\ref{eqlemma2}%
) implies (\ref{eqlemmahyp4})
\end{proof}

\bigskip

\noindent PROPOSITION \ref{propwelfare} \begin{proof}
Suppose that $W(\mu_{R}, \mu_{G})$ is
maximized at $(\mu_{R}, \mu_{G})=(\tilde{\mu}_{R}, \tilde{\mu}_{G})$, where
$\tilde{\mu}_{R} \in(0,1)$ and $\tilde{\mu}_{G} \in(0,1)$. Define $c\equiv
L_{A}(\tilde{\mu}_{R}, \tilde{\mu}_{G})/L_{B}(\tilde{\mu}_{R}, \tilde{\mu}%
_{G})$, and consider the constrained maximization problem:
\begin{equation}
\label{eqmaxconstr}\max_{\mu_{R} \in[0,1], \mu_{G} \in[0,1]} W(\mu_{R},
\mu_{G}) \mbox{ s.t. } L_{A}(\mu_{R}, \mu_{G})=c L_{B}(\mu_{R}, \mu_{G}).
\end{equation}
Because by definition of $c$, the solution $(\tilde{\mu}_{R}, \tilde{\mu}%
_{G})$ satisfies the restriction
\begin{equation}
g(\mu_{R}, \mu_{G}) = c L_{B}(\mu_{R}, \mu_{G}) - L_{A}(\mu_{R}, \mu_{G})=0,
\label{eqconstr}%
\end{equation}
it actually solves the maximization problem (\ref{eqmaxconstr}).
Define the feasible set $C = \{\mu_{R} \in[0,1], \mu_{G} \in[0,1] | g(\mu_{R},
\mu_{G})=0\}$. By the assumption of constant returns to scale, we have that
for all $(\mu_{R}, \mu_{G})\in C$: $w_{A}(\mu_{R}, \mu_{G})$ and $w_{B}%
(\mu_{R}, \mu_{G})$ are constant, and therefore, at all $(\mu_{R}, \mu_{G})\in
C$, the welfare function (\ref{eqwelfare2}) can be written as
\[
W(\mu_{R}, \mu_{G}) = L_{A}(\mu_{R}, \mu_{G})(U(\overline{w_{A}})+
U(\overline{w_{B}})/c),
\]
which is monotone increasing with $L_{A}(\mu_{R}, \mu_{G})$. Therefore, the
solution $(\tilde{\mu}_{R}, \tilde{\mu}_{G})$ also solves the following
maximization problem:
\begin{equation}
\label{eqmaxconstr2}\max_{\mu_{R} \in[0,1], \mu_{G} \in[0,1]} L_{A}(\mu_{R},
\mu_{G}) \mbox{ s.t. } L_{A}(\mu_{R}, \mu_{G})=c L_{B}(\mu_{R}, \mu_{G}).
\end{equation}
We verify that $(\tilde{\mu}_{R}, \tilde{\mu}_{G})$ indeed satisfy the first-
and second-order conditions of problem (\ref{eqmaxconstr2}). The Lagrangian is
given by
\[
\mathcal{L}(\mu_{R}, \mu_{G},\psi) = (1 - \psi)L_{A}(\mu_{R}, \mu_{G})+\psi c
L_{B}(\mu_{R}, \mu_{G}).
\]
Since $(\tilde{\mu}_{R}, \tilde{\mu}_{G})$ is supposed to be interior, the
following first order constraints should hold:
\begin{align}
\label{eqfoc2}\frac{\partial\mathcal{L}}{\partial\mu_{R}}(\tilde{\mu}_{R},
\tilde{\mu}_{G},\psi) = (1 - \psi)\frac{\partial L_{A}}{\partial\mu_{R}%
}(\tilde{\mu}_{R}, \tilde{\mu}_{G}) + \psi c\frac{\partial L_{B}}{\partial
\mu_{R}}(\tilde{\mu}_{R}, \tilde{\mu}_{G})  &  = 0\\
\frac{\partial\mathcal{L}}{\partial\mu_{G}}(\tilde{\mu}_{R}, \tilde{\mu}%
_{G},\psi) = (1 - \psi)\frac{\partial L_{A}}{\partial\mu_{G}}(\tilde{\mu}_{R},
\tilde{\mu}_{G}) + \psi c\frac{\partial L_{B}}{\partial\mu_{G}}(\tilde{\mu}_{R},
\tilde{\mu}_{G})  &  = 0.
\end{align}
The first part of Lemma~\ref{lemma2} implies that $\psi\in(0,1)$ and that
under condition (\ref{eqlemma2}): $\mu_{R}>\mu_{G}$ if and only if
$\partial\mathcal{L}/\partial\mu_{R} > \partial\mathcal{L}/\partial\mu_{G}$.
Therefore, condition (\ref{eqlemma2}) and the first-order conditions imply
that $\tilde{\mu}_{R}=\tilde{\mu}_{G}\equiv\tilde{\mu}$.
Since $\tilde{\mu}_{R}=\tilde{\mu}_{G}$ defines a unique point in $C$, the
second-order condition should hold at $\tilde{\mu}_{R}=\tilde{\mu}_{G}$, which
says that the Hessian of the Lagrangian with respect to $(\mu_{R},\mu_{G})$
evaluated at the social optimum, $D_{\mu_{R},\mu_{G}}^{2}\mathcal{L}%
(\tilde{\mu},\tilde{\mu},\psi)$, is negative definite on the subspace
$\{z_{R},z_{G}|z_{R}(\partial g/\partial\mu_{R})+z_{G}(\partial g/\partial
\mu_{G})=0\}$. The second order condition is thus that at $(\mu_{R},\mu
_{G})=(\tilde{\mu},\tilde{\mu})$:
\begin{equation}
2\frac{\partial g}{\partial\mu_{R}}\frac{\partial g}{\partial\mu_{G}}%
\frac{\partial^{2}\mathcal{L}}{\partial\mu_{R}\partial\mu_{G}}-\left(
\frac{\partial g}{\partial\mu_{R}}\right)  ^{2}\frac{\partial^{2}\mathcal{L}%
}{(\partial\mu_{G})^{2}}-\left(  \frac{\partial g}{\partial\mu_{G}}\right)
^{2}\frac{\partial^{2}\mathcal{L}}{(\partial\mu_{R})^{2}}>0. \label{eqsoc}%
\end{equation}
Because $\frac{\partial g}{\partial\mu_{R}}(\tilde{\mu},\tilde{\mu}%
)=\frac{\partial g}{\partial\mu_{G}}(\tilde{\mu},\tilde{\mu}),$ and
$\frac{\partial^{2}\mathcal{L}}{(\partial\mu_{G})^{2}}(\tilde{\mu},\tilde{\mu
})=\frac{\partial^{2}\mathcal{L}}{(\partial\mu_{R})^{2}}(\tilde{\mu}%
,\tilde{\mu}),$ the second order condition (\ref{eqsoc}) simplifies to
$\frac{\partial^{2}\mathcal{L}}{\partial\mu_{R}\partial\mu_{G}}(\tilde{\mu
},\tilde{\mu})>\frac{\partial^{2}\mathcal{L}}{(\partial\mu_{R})^{2}}%
(\tilde{\mu},\tilde{\mu}),$ or equivalently
\begin{equation}
(1-\psi)\frac{\partial^{2}L_{A}}{\partial\mu_{R}\partial\mu_{G}}(\tilde{\mu
},\tilde{\mu})+\psi c \frac{\partial^{2}L_{B}}{\partial\mu_{R}\partial\mu_{G}%
}(\tilde{\mu},\tilde{\mu})>(1-\psi)\frac{\partial^{2}L_{A}}{(\partial\mu
_{R})^{2}}(\tilde{\mu},\tilde{\mu})+\psi c\frac{\partial^{2}L_{B}}{(\partial
\mu_{R})^{2}}(\tilde{\mu},\tilde{\mu}). \label{eqsoc2}%
\end{equation}
By the second part of Lemma \ref{lemma2}, inequality (\ref{eqsoc2}) cannot
hold under condition (\ref{eqlemma2}). Therefore we have a contradiction and
the non-segregation allocation $(\tilde{\mu}_{R},\tilde{\mu}_{G})$ cannot be a
social optimum. Since a social optimum exists by continuity of $W$ and
compactness of $[0,1]^{2}$, the social optimum necessarily has to involve
complete or partial segregation.
\end{proof}

\bibstyle{aea}

\end{document}